%% file: MT13-Achievable-Performance.tex
\documentclass[10pt,reqno,twoside,onecolumn]{article}
\input{macro.tex}
\author{Michael B.\ McCoy and Joel A.\ Tropp}
\date{September 28, 2013}
\title{The achievable performance of convex demixing}

\begin{document} \maketitle 
\begin{abstract} 
 Demixing is the problem of identifying multiple structured signals from a superimposed, undersampled, and noisy observation. This work analyzes a general framework, based on convex optimization, for solving demixing problems.  When the constituent signals follow a generic incoherence model, this analysis leads to precise recovery guarantees.  These results admit an attractive interpretation: each signal possesses an intrinsic degrees-of-freedom parameter, and demixing can succeed if and only if the dimension of the observation exceeds the total degrees of freedom present in the observation. 
\end{abstract}

\section{Introduction}
\label{sec:introduction}

\Lettrine{D}{emixing} refers to the problem of extracting multiple informative signals from a single, possibly noisy and undersampled, observation.  One rather general model for a mixed observation \(\vct z_0 \in \R^m\) takes the form
\begin{equation}\label{eq:obs-model}
  \vct z_0 = \mtx A \left (\sumnl_{i=1}^n \mtx U_i\vct x_i^\nat + \vct w \right),
\end{equation}
where the constituents \((\vct x_i^\nat)_{i=1}^n\) are the unknown informative signals that we wish to find;  the  matrices \((\mtx U_i)_{i=1}^n\)  model the relative orientation of the constituent vectors; the  operator \(\mtx A \in \R^{m\times d}\) compresses the observation from \(d\) dimensions to \(m\le d\) dimensions;  and
 \(\vct w \in \R^d\) is unstructured noise.  We  assume that all elements appearing in~\eqref{eq:obs-model} are known except for  the constituents \((\vct x_i^\nat)_{i=1}^n\) and the noise \(\vct w\).

 Numerous applications of the model~\eqref{eq:obs-model} appear in modern data-intensive science. In imaging, for example, the informative signals can model features like stars and galaxies~\cite{StaDonCan:03}, while an undersampling operator accounts for known occlusions or missing data~\cite{ElaStaQue:05a,StuKupPop:12}. In graphical model selection, the data may consist of the sum of a sparse component that encodes causality structure and a confounding low-rank component that arises from unobserved latent variables~\cite{ChaParWil:10}. Similar mixed-signal models appear in robust statistics~\cite{CanLiMa:11,CheJalSan:13} and image processing~\cite{PenGanWri:12,WriGanMin:13}.   In every case, the  question of interest is
 \begin{equation*}
   \text{\emph{When is it possible to recover the constituents from the observation?}}
 \end{equation*}

This work answers  this question for a popular class of demixing procedures under a random model.  The analysis reveals that each constituent possesses a degrees-of-freedom parameter, and that these demixing procedures can succeed with high probability if and only if the total  number of measurements exceeds the total degrees of freedom.

In the next two subsections, we describe a well-known recipe that  converts \emph{a priori} structural information on the constituents \(\vct x_i^\nat\) into an convex optimization program suited for demixing~\eqref{eq:obs-model}.  Section~\ref{sec:rand-orient-model} motivates a random model that we use to study demixing, and  Section~\ref{sec:stat-dimens}  defines the degrees-of-freedom parameter \(\delta_i\).  The main result  appears in Section~\ref{sec:main-results-1}.

\subsection{Structured signals and  convex penalties}
\label{sec:struct-incoh}

In the absence of assumptions, it is impossible to reliably recover unknown vectors from a superposition of the form~\eqref{eq:obs-model}.  In order to have any hope of success, we must make use of domain-specific knowledge about the types of constituents making up our observation.  This knowledge often implies that our constituents belong to some set of highly-structured elements. Typical examples of these structured families include \emph{sparse vectors} and \emph{low-rank matrices}.
\begin{description}\setlength{\itemsep}{-2pt}
\item[\sffamily Sparse vectors]  A sparse vector has many entries equal to zero.  Sparse vectors regularly appear  in modern signal and data processing applications for a variety of reasons. Bandlimited communications signals, for example, are engineered to be sparse in the frequency domain.   The adjacency matrix of a sparse graph is sparse by definition. Piecewise smooth functions are nearly sparse in wavelet bases, so that many natural images exhibit sparsity in the wavelet domain~\cite[Sec.~9]{Mal:09}.  
\item[\sffamily Low-rank matrices] A matrix has low rank if many of its singular values are equal to zero.  Low-rank structure appears whenever the rows or columns of a matrix satisfy many nontrivial linear relationships. For example, strong correlations between predictors cause many statistical datasets to exhibit low-rank structure. Rank deficient matrices appear in a number of other areas,  including control theory~\cite[Sec.~6]{Faz:02},  video processing~\cite{CanLiMa:11},  and structured images~\cite{PenGanWri:12}.
\end{description}
Other types of structured families that appear in the literature include the family of sign vectors \(\{\pm 1\}^d \subset \R^d\)~\cite{ManRec:11},  nonnegative sparse vectors~\cite{DonTan:10}, block- and group-sparse vectors and matrices~\cite{RaoKre:98,MalCetWil:03}, and orthogonal matrices~\cite{ChaRecPar:12}.

In each of these cases, the structured family possesses an associated convex function that, roughly speaking, measures the amount of complexity of a signal with respect to the family~\cite{DeVTem:96,Tem:03,ChaRecPar:12}.  For sparse vectors and low-rank matrices, the natural penalty functions are the \(\ell_1\) norm and the Schatten 1-norm:
\begin{equation*}
  \lone{\vct x} := \sumnl_{i=1}^d\abs{x_i} \qtq{and} \sone{\vct X} := \sumnl_{i=1}^{p\wedge q} \sigma_i(\vct X),
\end{equation*}
where \(\sigma_i(\vct X)\) is the \(i\)th singular value of \(\vct X\) and the wedge \(\wedge\) denotes the minimum of two numbers. See~\cite[Sec.~2.2]{ChaRecPar:12} for additional examples as well as a principled approach to constructing  convex penalty functions.  These convex complexity measures form the building blocks of the demixing procedures that we study in this work. 

\subsection{A generic demixing framework}
\label{sec:gener-demix-fram}

Given an observation of the form~\eqref{eq:obs-model}, we desire a computational method for recovering  the constituents \(\vct x_i^\nat\).   We now describe a well-known framework that combines convex complexity measures into a convex optimization program that demixes a signal.  Specific instances of this recipe appear in numerous works~\cite{DonHuo:01,ChaSanPar:09,CheJalSan:13,PenGanWri:12}, and the general format described below is closely related to the work~\cite{McCTro:12,WriGanMin:13}. 

Assume that, for each constituent \(\vct x_i^\nat\), we have determined an appropriate convex complexity function \(f_i\).  For example, if  we suspect that the \(i\)th constituent \(\vct x_i^\nat\) is sparse, we may choose the \(f_i = \lone{\cdot}\), the \(\ell_1\) norm.   In the \emph{Lagrange formulation} of the demixing procedure, we combine the regularizers into a single master penalty function \(F_{\vct \lambda} \colon (\R^d)^n\to \R\) given by
\begin{equation*}
  F_{\vct \lambda}(\vct x_1,\dotsc,\vct x_{n-1},\vct x_n):= \sumnl_{i=1}^n \lambda_i f_i(\vct x_i),
\end{equation*}
where the weights  \(\lambda_i > 0\).  In this formulation, we minimize the master penalty \(F_{\vct \lambda}\) plus a Euclidean-norm penalty constraint that ensures consistency with our observation:
\begin{equation}\label{eq:noisy-demix-lag}
  \minimizeOp_{\vct x_i \in \R^d}\; F_{\vct \lambda}(\vct x_1,\dotsc, \vct x_{n-1},\vct x_n) +
  \enormsq{\mtx A^\psinv \left[ \mtx A \left(\sumnl_{i=1}^n\mtx U_i \vct x_i\right) - \vct z_0\right]}
\end{equation}
where \(\enormsq{\vct x}:= \langle \vct x,\vct x\rangle\) is the squared Euclidean norm.  We include the Moore--Penrose pseudoinverse \(\mtx A^\dagger\) in the consistency term to ensure that our recovery procedure is independent of the conditioning of \(\mtx A\).   This demixing procedure succeeds when an optimal point \((\tvct x_i)_{i=1}^n\) of~\eqref{eq:noisy-demix-lag}  provides a good approximation for the true constituents \((\vct x_i^\nat)_{i=1}^n\).  

Rather than restrict ourselves to specific choices of Lagrange parameters \(\vct \lambda\), we study whether it is possible to demix the constituents of \(\vct z_0\) using a method of the form~\eqref{eq:noisy-demix-lag} for the best choice of weights \(\vct\lambda\). To study this setting, we focus our analysis on the more powerful \emph{constrained formulation} of demixing:
\begin{equation}\label{eq:const-dmx}
  \begin{aligned}
    \minimizeOp_{\vct x_i \in \R^d}\quad &\enormsq{\mtx A^\psinv \left(\mtx A \sumnl_{i=1}^n \mtx U_i\vct x_i - \vct z_0 \right)}  \\ \subjectto& f_i(\vct x_i) \le f_i(\vct x_i^\nat) \qtq{for each} i =
    1,\dotsc,n-1,n.
 \end{aligned}
\end{equation}
The theory of Lagrange multipliers indicates that solving the constrained demixing program~\eqref{eq:const-dmx} is essentially equivalent to solving the Lagrange problem~\eqref{eq:noisy-demix-lag} with the best choice of weights \(\vct \lambda\).  There are some subtle issues in this equivalence, notably the fact that~\eqref{eq:noisy-demix-lag} can have strictly more optimal points than the corresponding constrained problem~\eqref{eq:const-dmx}.  We refer to~\cite[Sec.~28]{Roc:70} for further details.

We wish to interrogate whether an optimal point \((\hvct x_i)_{i=1}^n\) of~\eqref{eq:const-dmx} forms a good approximation for the true constituents \((\vct x_i^\nat)_{i=1}^n\).  For this study, we distinguish two situations.
\begin{description}\setlength{\itemsep}{-1pt}
\item[Exact recovery] In the noiseless setting where \(\vct w = \zerovct\), can we guarantee that the constrained demixing program~\eqref{eq:const-dmx} recovers the constituents exactly?
\item[Stable recovery] For nonzero noise \(\vct w \ne \zerovct\), can we guarantee that any solution to the constrained demixing problem~\eqref{eq:const-dmx} provides a good approximation to the  constituents \(\vct x_i^\nat\)?
\end{description}
The following definition makes these notions precise.  
\begin{definition}[Exact and stable recovery]  \label{def:exact-stable}
  We say that \emph{exact recovery is achievable} in~\eqref{eq:const-dmx} if the tuple \((\vct x_i^\nat)_{i=1}^n\) is the unique optimal point of~\eqref{eq:const-dmx} when \(\vct w = \zerovct\).  We say that \emph{stable recovery is achievable} if there exists number \(C>0\), such that for any optimal point \((\hvct x_i)_{i=1}^n\) of~\eqref{eq:const-dmx}, we have
  \begin{equation}\label{eq:stab-recov-defn}
    \enorm{\hvct x_i -\vct x_i^\nat} \le C \enorm{\vct w} \qtq{for all} i = 1,\dotsc,n-1,n.
  \end{equation}
  The value of \(C\)  may depend on all  problem parameters except \(\vct w\).
\end{definition}
The goal of this work is to describe when exact and stable recovery are achievable for the constrained demixing program~\eqref{eq:const-dmx}.

\subsection{A generic model for incoherence}
\label{sec:rand-orient-model}

A necessary requirement to identify signals from a superimposed observation is that the constituent signals must look different.  The superposition of two sparse vectors, for example, is still sparse; \emph{a priori}  knowledge that both vectors are sparse provides little guidance in determining how to allocate the nonzero elements between the two constituents.  On the other hand, a sparse vector looks very different from a superposition of a small number of sinusoids.  This structural diversity makes distinguishing spikes from sines tractable~\cite{Tro:08}.   We extend this idea to more general families by saying that structured vectors that look very different from one another are \emph{incoherent}.  

 In this work, we follow~\cite{DonHuo:01,McCTro:12} and model incoherence by assuming that the families are \emph{randomly oriented} relative to one another.   The set of all possible orientations on \(\R^d\) is the \emph{orthogonal group} \(\orth{d}\) consisting of all \(d\times d\) orthogonal matrices:
\begin{equation*}
  \orth{d} :=   \bigl\{\mtx U \in \R^{d\times d} \mid \mtx U^\transp \mtx U = \Id\bigr\}.
\end{equation*}
The orthogonal group is a compact group, and so it possesses a unique invariant probability measure called the \emph{Haar measure}~\cite[Ch.~44]{Fre:03}.  We model incoherence among the constituents \(\vct x_i^\nat\) by drawing the orientations \(\mtx U_i\) from the Haar measure. 
\begin{definition}[Random orientation model]  
  We say that the matrices \((\mtx U_i)_{i=1}^d\) satisfy the \emph{random orientation model} if the matrices \(\mtx U_1, \dotsc, \mtx U_{n-1},\mtx U_n\)  are drawn independently from the Haar measure on~\(\orth{d}\).  
\end{definition}

The random orientation model is analogous to  random measurements models that appear in the compressed sensing literature~\cite{CanTao:05,Don:06}.  In this work, however, we find that orienting the structures randomly through the rotations \(\mtx U_i\) provides sufficient randomness for the analysis. We have no need to assume that the measurement matrix \(\mtx A\) is random.

\subsection{Descent cones and the statistical dimension}
\label{sec:stat-dimens}

Our study of the exact and stable recovery capabilities of the constrained demixing program~\eqref{eq:const-dmx} relies on a geometric analysis of the optimality conditions of the convex program~\eqref{eq:const-dmx}.  The key player in this analysis is the following  cone that captures the local behavior of a convex function at a point (Figure~\ref{fig:descent}).
\begin{figure}[t!]
  \centering
  \includegraphics[width=0.4\columnwidth]{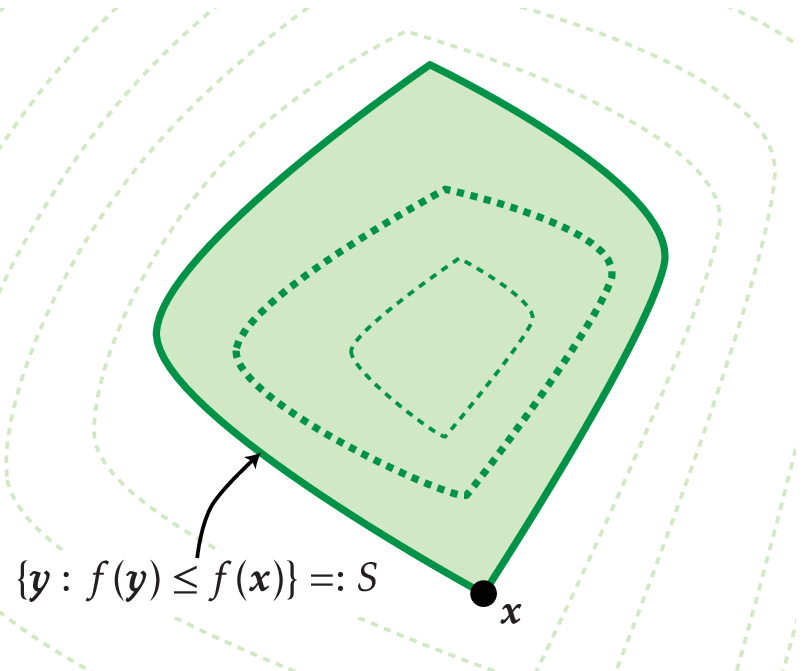}
  \hspace{1cm}
  \includegraphics[width=0.4\columnwidth]{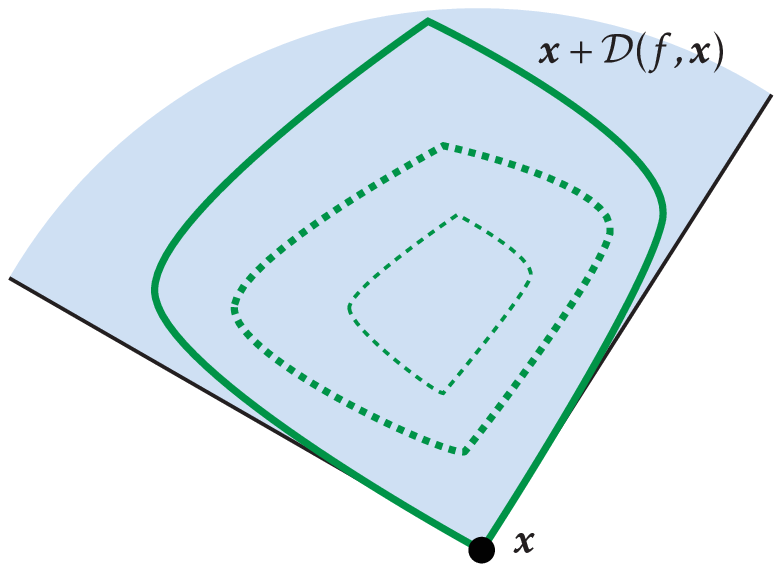}
  \caption[Descent cone.]{{\sffamily Descent cone.} { \sf [Left]} The sublevel set \(S\) {\it(shaded)} of a convex function \(f\) {\it(level lines)} at a point \(\vct x\). { \sf [Right]} The descent cone \(\Desc(f,\vct x)\) {\it(shaded)} is the cone generated by  \(S\) at \(\vct x\). }
  \label{fig:descent}
\end{figure}
\begin{definition}[Descent cone]
  The descent cone \(\Desc(f,\vct x)\) of a convex function \(f\) at a point \(\vct x\) is the cone generated by the perturbations about \(\vct x\) that do not increase \(f\):
  \begin{equation}\label{eq:desc-cone-def}
    \Desc(f,\vct x) :=  \{\vct y \mid f(\vct x + \tau \vct y)\le f(\vct x) \text{ for some }\tau >0\}.
  \end{equation}
\end{definition}
\noindent Intuitively, a convex penalty function \(f\) will be more effective at finding a structured vector \(\vct x^\nat\) if most perturbations around \(\vct x^\nat\) increase the value of \(f\), i.e., if the descent cone \(\Desc(f,\vct x^\nat)\) is small.  Our next definition provides a summary parameter that lets us quantify the size of a convex cone.
\begin{definition}[Statistical dimension]
  Let \( C \subset \R^d\) be a closed convex cone, and define the Euclidean projection \(\Proj_C\colon \R^d \to C\)  onto \(C\) by
  \begin{equation*}
    \Proj_{C}(\vct x) := \argmin_{\vct y \in C} \enormsq{\vct y-\vct x}.
  \end{equation*}
  The \emph{statistical dimension} \(\sdim(C)\) of \(C\) is given by the average value
  \begin{equation}\label{eq:sdim-def}
    \sdim(C) := \Expect_{\vct g}\left[ \enormsq{\Proj_C(\vct g)}\right],
  \end{equation}
  where \(\vct g\sim \normal(\zerovct,\Id)\) is a standard Gaussian vector. %
\end{definition}
The statistical dimension satisfies a number of properties that make it an appropriate measure of the ``size'' or ``dimension'' of a convex cone.  It also  extends a number of useful properties for the usual dimension of  a linear subspace to convex cones~\cite[Sec.~4]{AmeLotMcC:13}.  Moreover, a number of  calculations for the statistical dimension are available in the literature~\cite{StoParHas:09,ChaRecPar:12,AmeLotMcC:13,FoyMac:13}, which makes the statistical dimension an appealing parameter in practice.  

The statistical dimension turns out to be the key parameter which determines the success and failure of  demixing under the random orientation model.  To shorten notation, we abbreviate the statistical dimensions of the descent cones \(\Desc(f_i,\vct x_i^\nat)\):
\begin{equation}\label{eq:sdim_i}
  \delta_i := \sdim_i\Bigl(\overline{\Desc(f_i,\vct x_i^\nat)}\Bigr) \qtq{for} i = 1,\dotsc,n-1,n,
\end{equation}
where the overline denotes the closure.  

\subsection{Main result}
\label{sec:main-results-1}

We are now in a position to state our main result. %
\begin{bigthm}\label{thm:mult-bd}
  With \(\sdim_i\) as in~\eqref{eq:sdim_i}, define the total  dimension \(\Delta\) and the scale  \(\sigma\) by
  \begin{equation}\label{eq:Delta-sigma}
    \Delta := \sumnl_{i=1}^n \delta_i
    \qtq{and}    
    \sigma := 
    \sqrt{\sumnl_{i=1}^n \delta_i\wedge (d-\delta_i)}.
  \end{equation}
  Choose a  probability tolerance \(\eta \in (0,1)\), and define the transition width
  \begin{equation}\label{eq:lambda-star}
    \lambda_* := \tfrac{4}{3} \log\left(\tfrac{1}{\eta}\right)+ 2\sigma\sqrt{ \log\left(\tfrac{1}{\eta} \right)} .
  \end{equation}
Suppose that the matrices \((\mtx U_i)_{i=1}^n\)  are drawn from the random orientation model and that the measurement operator \(\mtx A \in \R^{m\times n}\) has full row rank.  Then
  \begin{align}
    m \ge \Delta + \lambda_* &\implies \text{stability is achievable with probability} \ge 1-\eta;\label{eq:success-mult-demix}
    \\
    m \le \Delta - \lambda_* &\implies  \text{exact recovery is achievable with probability} \le \eta, \label{eq:fail-mult-demix}
  \end{align}
  where we define exact and stable recovery  in Definition~\ref{def:exact-stable}.
\end{bigthm}
\noindent
Theorem~\ref{thm:mult-bd} provides detailed information about the capability of constrained demixing~\eqref{eq:const-dmx} under the random orientation model. 
\begin{description}\setlength{\itemsep}{-1pt}
\item[Phase transition]  The capability of~\eqref{eq:const-dmx} changes rapidly when the number of measurements \(m\) passes through the total statistical dimension \(\Delta\).  For  \(m\) somewhat less than \(\Delta\), exact recovery is highly unlikely.  On the other hand, when  \(m\) is a bit larger than \(\Delta\), we have stable recovery with high probability.  This justifies our heuristic that the number of measurements required for demixing is equal to the total statistical dimension.  %
\item[Transition width]  Theorem~\ref{thm:mult-bd} tightly controls the width of the transition region between success  and failure.  When the  probability tolerance \(\eta\) is independent of \(d\) and \(n\), the transition width satisfies
  \begin{equation}\label{eq:lam-order}
\lambda_* =  O(\sigma) = O\bigl(\sqrt{nd}\bigr) \qtq{as} d\to \infty.
\end{equation}
The second equality follows from the observation that \(\sigma^2 \le nd\) because the statistical dimension is never larger than the ambient dimension (cf.~\cite[Sec.~4]{AmeLotMcC:13}).  In many applications, the number \(n\) of constituents is independent of the ambient dimension, so the transition between success and failure occurs over no more than \(O(\sqrt{d})\) measurements as \(d\to\infty\).  
\item[Strong probability bounds] Probability tolerances \(\eta\) that decay rapidly with the ambient dimension \(d\) can provide strong guarantees for demixing~\cite[Sec.~4.3]{McCTro:12}.  For example, when the number of measurements \(m \ge \Delta + c d\) for some \(c>0\), Theorem~\ref{thm:mult-bd} guarantees that
  \begin{equation*}
     \text{\emph{Stability is achievable with probability}} \ge 1-\econst^{-c'd}.
  \end{equation*}
Due to the estimate  \(\sigma^2 \le nd\) from above, the constant \(c' > 0\) need depend only on \(c\) and \(n\).  Such exponentially small failure probabilities lead to strong demixing bounds using union-bound arguments as in~\cite[Secs.~6.1.1 \& 6.2.2]{McCTro:12}.  We omit the details for brevity.
\item[Extreme demixing]  How many constituents can we reliably demix?  The  answer is simple:
  \begin{equation*}
    \text{\emph{Theorem~\ref{thm:mult-bd} allows \(n\) proportional to \(d\).}}
  \end{equation*}
  Consider, for example,  the fully observed case \(m = d\), and fix a probability of success \(\eta\) independent of \(d\).  Suppose that \(\Delta \le (1-\eps)d \) for some \(\eps >0\).  For demixing to succeed, by Theorem~\ref{thm:mult-bd}, we only need
  \begin{equation*}
     d - \Delta   \ge \eps  d \ge \lambda_* =  O(\sqrt{nd}) \qtq{as} d\to \infty
  \end{equation*}
 where the equality is~\eqref{eq:lam-order}. Thus,  the implication~\eqref{eq:success-mult-demix} remains nontrivial  as~\(d\to \infty\) so long as  \(n \le cd\) for some sufficiently small \(c>0\).  %

 This growth regime is essentially optimal.  It can be shown\footnote{The fact that \(\delta_i>1/2\) except in trivial cases follows because (1) the statistical dimension of a ray is \(0.5\), (2) every nontrivial cone contains a ray,  and (3) the statistical dimension is increasing under set inclusion.} that \(\delta_i \ge 1/2\) whenever \(\vct x_i^\nat\) is not the unique global minimum of \(f_i\). Thus, excepting trivial situations, we have \(\Delta \ge n/2\), so that when \(n \gtrsim 2 d\), demixing must fail with high probability by~\eqref{eq:fail-mult-demix}.
\end{description}

The proof of Theorem~\ref{thm:mult-bd} is based on a geometric optimality condition for the constrained demixing program~\eqref{eq:const-dmx} that characterizes exact and stable recovery in terms of  a configuration of randomly oriented convex cones. A new extension of the approximate kinematic formula from~\cite{AmeLotMcC:13} lets us provide precise bounds on the probability that this geometric optimality condition holds under the random orientation model.

\subsection{Outline}
\label{sec:outline}

Section~\ref{sec:related-work} describes the related work on demixing.   The proof of Theorem~\ref{thm:mult-bd} appears in Section~\ref{sec:constrained-demixing}.  Section~\ref{sec:applications} provides two simple numerical experiments that illustrate the accuracy of Theorem~\ref{thm:mult-bd}, and we conclude in Section~\ref{sec:open-problems} with some open problems.  The technical details in our development appear in the appendices.

\subsection{Notation and basic facts}
\label{sec:notation}
Vectors appear in bold  lowercase, while matrices are bold and capitalized.      The range and nullspace of a matrix \(\mtx X\) are \(\range(\mtx X)\) and \(\nullity(\mtx X)\).  The  Minkowski sum of sets \(S_1,S_2 \subset \R^d\)  is \(S_1+S_2\).  When more than two sets are involved, we define the Minkowski sum \(\sum_{i} S_i\) inductively. We write \(-S\) for the reflection of \(S\) about \(\zerovct\) and  \(\overline{S}\) for the closure of  \(S\).

A convex cone \(C\subset \R^d\) is a convex set that is positive homogeneous: \(\vct x, \vct y \in C\implies \lambda(\vct x+\vct y)\in C\) for all \(\lambda >0\).  All cones  in this work contain the origin \(\zerovct\).  We write \(\cC_d\) for the set of all closed, convex cones in \(\R^d\).   For any cone \(C \subset \R^d\), we define the polar cone \(C^\polar \in  \cC_d\) by
\begin{equation}      \label{eq:polar-cone}
  C^\polar := 
  \{\vct y \in \R^d \mid \langle \vct x,\vct y\rangle \le 0 \qtq{for all} \vct x \in C\}.  
\end{equation}
The bipolar formula states \(C^{\polar \polar} = \overline C\).   We measure the distance between two cones \(C,D \subset \R^d\) by computing the maximal inner product
\begin{equation}\label{eq:conic-dist}
  \llangle C,D\rrangle := 
  \sup_{\substack{\vct x \in C\cap \ball{d}\\ \vct y \in D\cap \ball{d}}} \langle \vct x,\vct y\rangle.
\end{equation}
It follows from the  Cauchy--Schwarz inequality that \(\llangle C, D\rrangle \le 1\) for every pair of cones, while the equality conditions for  Cauchy--Schwarz show that \(\llangle C,D \rrangle = 1\)    if and only if the intersection \(\overline C \cap \overline D\) contains a ray. 

We will refer to the following  elementary properties of the statistical dimension.  For any closed convex cones \(C \in \cC_d\) and \(D \in \cC_{d'}\), the statistical dimension reverses under polarity
\begin{equation}
  \label{eq:sdim-polar}
  \sdim(C^\polar) = d-\sdim(C)
\end{equation}
and splits under the Cartesian product
\begin{equation}
  \label{eq:sdim-prod}
  \sdim(C\times D) = \sdim(C) + \sdim(D).
\end{equation}
Simple proofs of relations~\eqref{eq:sdim-polar} and~\eqref{eq:sdim-prod} appear in~\cite[Sec.~4]{AmeLotMcC:13}.

\section{Context and related work}
\label{sec:related-work}

This work is a successor to  the author's earlier work~\cite{McCTro:12} on demixing with \(n=2\) components  in the fully observed \(m=d\) setting.   The techniques used in this paper hail from~\cite{AmeLotMcC:13}, which studied phase transitions in randomized optimization programs.  While those two works are the closest in spirit to our development below, numerous works on demixing appear in the literature.  This section provides an overview of the literature on demixing, from its origins in sparse approximation to recent developments towards a general theory.

\paragraph*{Demixing and sparse approximation.}
\label{sec:demix-sparse-appr}

Early work on demixing methods used the  \(\ell_1\) norm to encourage sparsity. Taylor, Banks, \& McCoy~\cite{TayBanMcC:79} used~\eqref{eq:noisy-demix-lag} with \(f_1=f_2=\lone{\cdot}\)  to demix a sparse signal from  sparse noise, with  applications to geophysics. About ten years later, Donoho \& Stark~\cite{DonSta:89} explained how uncertainty principles can guarantee the success of demixing signals that are sparse in frequency from those that are sparse in time using the \(\ell_1\) norm.  

The analysis of Donoho \& Huo~\cite{DonHuo:01} provided incoherence-based guarantees which demonstrate that exact recovery is possible under fairly generic conditions.   This work motivated interest in \emph{morphological component analysis} (MCA) for image processing~\cite{StaDonCan:03,ElaStaQue:05,BobMouSta:06,BobStaFad:07a,BobStaFad:07}.  MCA posits that images are the superposition of a small number of signals from a known dictionary---such as pointillistic stars and wispy galaxies. Demixing with the \(\ell_1\) norm provides a computational framework for decomposing these images into their constituent signals. 

 A number of recent papers provide theoretical guarantees for demixing with the \(\ell_1\) norm. Wright \& Ma showed that \(\ell_1\)-norm demixing can recovery a nearly dense vector from a sufficiently sparse corruption~\cite{WriMa:09}. Additional work along these lines appears in~\cite{StuKupPop:12,PopBraStu:13,NguTra:13,Li:13}. The phase transition for demixing two signals using the \(\ell_1\)-norm was first identified by the present authors in~\cite{McCTro:12}.   The very recent work~\cite{FoyMac:13} recovers similar guarantees under a slightly different model, and it also provides  stability guarantees.

\paragraph*{Demixing beyond sparsity.}
\label{sec:moving-beyond-spars}

Applications for mixed signal model~\eqref{eq:obs-model}  when the constituents satisfy more general structural assumptions appear in a number of areas.  The  work of Chandrasekaran et al.~\cite{ChaSanPar:09,ChaParWil:10,ChaSanPar:11}  demonstrated that a demixing program of the form~\eqref{eq:noisy-demix-lag} can recover the superposition of a sparse and low-rank matrix.  The independent work of Cand\`es et al.~\cite{CanLiMa:11} uses this model for robust principal component analysis and image processing applications. 

Modifications to the rank-sparsity model find applications in robust statistics~\cite{XuCarSan:10a,McCTro:11,XuCarSan:12} and its compressed variants~~\cite{WriGanMin:13,WatSanBar:11},  image processing~\cite{WriGanMin:13,PenGanWri:12}, and network analysis~\cite{JalRavSan:10,JalRavSan:11,CheJalSan:11,CheJalSan:11b,CheJalSan:13}.  

\paragraph*{A general theory takes shape.}
\label{sec:general-theory-takes}

Recent work has started to unify the piecemeal results discussed above.  Chandrasekaran et.\ al~\cite{ChaRecPar:12} gave a  general treatment of  the  \(n=1\) case using Gaussian width analysis. For the \(n=2\) and \(m=d\) case, the present authors used tools from integral geometry to demonstrate numerically matching upper and lower exact recovery guarantees for demixing~\cite{McCTro:12}.  The first fully rigorous account of phase transitions in demixing problems, for the \(n=2\) and \(m=d\) case, appeared in work of Amelunxen et al.~\cite{AmeLotMcC:13}.  

In very recent work, Foygel \& Mackey~\cite{FoyMac:13} studied the \(n=2\) case with a linear undersampling model that differs slightly from the one we consider in this work.  These empirically sharp results recover and extend some of the bounds in~\cite{McCTro:12}, but they do not prove that a phase transition occurs.  Notably, the work of Foygel \& Mackey offers  guidance on the choice of Lagrange parameters. 

The only previous result for the demixing setup where the number of constituents \(n\) is arbitrary appears in Wright et al.~\cite{WriGanMin:13}. Their results provide recovery guarantees for the Lagrange formulation of the undersampled demixing program~\eqref{eq:noisy-demix-lag}  when sufficiently strong guarantees are available for the fully observed \(m=d\) case.  Their guarantees, however, do not identify the phase transition between success and failure.

\section{Proof of the main result}
\label{sec:constrained-demixing}

This section presents the arc of the argument leading to Theorem~\ref{thm:mult-bd}, but it postpones the proof of  intermediate results to the appendices.  Section~\ref{sec:nonr-recov-cond} describes deterministic conditions for exact and stable recovery. In  Section~\ref{sec:two-simplifications-}, we provide simplifications for these deterministic conditions that hold almost surely under the random orientation model.  These simplifications reduce the recovery conditions to a single geometric condition involving the intersection of (polars of) randomly oriented descent cones.  

Our key tool, the \emph{approximate kinematic formula}, appears in Section~\ref{sec:appr-kinem-form}.  This formula bounds the probability that an arbitrary number of randomly oriented cones intersect in terms of the statistical dimension. It 
extends and refines a result of Amelunxen et al.~\cite[Thm.~7.1]{AmeLotMcC:13}. 
We complete the argument in Section~\ref{sec:completing-proof} by applying the kinematic formula to our simplified geometric recovery condition.

\subsection{Deterministic recovery conditions}
\label{sec:nonr-recov-cond}

We begin the proof of Theorem~\ref{thm:mult-bd} with deterministic conditions for exact recovery and stability for the constrained demixing problem~\eqref{eq:const-dmx}.  These conditions rephrase exact recovery and stability in terms of configurations of descent cones. In order to highlight the symmetries in these conditions, we first introduce some notation that we use throughout the proof . Define
\begin{equation}\label{eq:cone-defs}
  D_i := \Desc(f_i,\vct x_i^\nat) \qtq{for} i=1,\dotsc,n,\quad D_{n+1} := \nullity(\mtx A),
  \qtq{and} \mtx U_{n+1} := \Id.
\end{equation}
The \emph{exact recovery condition} is the event
\begin{equation}\label{eq:erc}
  -\mtx U_i D_i \cap \left(\sumnl_{j\ne i}\mtx U_j D_j\right)  = \{\zerovct\}\qtq{for all} i=1,\dotsc,n,n+1. \tag{ERC}
\end{equation}
In words, the exact recovery condition requires that no descent cone  shares a ray with the {sum}  of the other cones.  The \emph{stable recovery condition}  strengthens~\eqref{eq:erc} by requiring that the cones are separated by some positive angle:
\begin{equation}\label{eq:src}
  \left\llangle -\mtx U_i D_i, \sumnl_{j\ne i}\mtx U_j D_j\right\rrangle < 1 \qtq{for all} i = 1,\dotsc, n, n+1,\tag{SRC}
\end{equation}
where we recall the definition~\eqref{eq:conic-dist} of the inner product between cones.  These two conditions precisely characterize exact and stable recovery for constrained demixing~\eqref{eq:const-dmx}.

\begin{lemma}\label{lem:achievable-stab}
  Success is achievable in the noiseless case if and only if the exact recovery condition~\eqref{eq:erc} holds.  If the stable recovery condition condition~\eqref{eq:src} holds, then stability is achievable. 
\end{lemma}
\noindent We prove Lemma~\ref{lem:achievable-stab} in Appendix~\ref{sec:determ-cond}. The proof of exact recovery is based on a perturbative argument that extends the proof~\cite[Lem.~2.3]{McCTro:12} of the recovery conditions for demixing two signals. The stable recovery result follows similar lines.

\subsection{Three simplifications}
\label{sec:two-simplifications-}

Our goal in this work is the analysis of  demixing when the orientations are drawn independently from the Haar measure on the orthogonal group.  In this section, we describe some simplifications that arise from the fact that this  measure is \emph{invariant} and \emph{continuous}.   In the end, we reduce the problem of studying the exact and stable recovery conditions~\eqref{eq:erc} and~\eqref{eq:src} hold to the problem of studying  a single geometric question: \emph{What is the probability that \(n+1\) randomly oriented cones share a ray?} 

In Section~\ref{sec:reduct-rand-nullsp}, we show that we can replace the deterministic nullspace  \(\nullity(\mtx A)\)  with a randomly oriented \(d-m\) dimensional subspace,  which effectively randomizes the nullspace of the measurement operator \(\mtx A\). In Section~\ref{sec:stab-almost-equiv}, we find that~\eqref{eq:erc} and~\eqref{eq:src} are  equivalent under the random orientation model.  Finally,  we simplify the exact recovery condition~\eqref{eq:erc} in Section~\ref{sec:polar-exact-recov-2}.  

\subsubsection{Randomizing the nullspace }
\label{sec:reduct-rand-nullsp}

In definition~\eqref{eq:cone-defs}, we fix the rotation \(\mtx U_i=\Id\) in order to make the statement of the exact and stable recovery  conditions symmetric.  However, this symmetry is broken by the random orientation model because only \((\mtx U_i)_{i=1}^n\) are taken at random.  The next result restores this symmetry.
\begin{lemma}\label{lem:rand-nullspace}
  Suppose that \((\mtx U_i)_{i=1}^{n}\) are drawn from the random orientation model and fix \(\mtx U_{n+1} = \Id\).  Let \((\mtx Q_i)_{i=1}^{n+1}\) be an \((n+1)\)-tuple of i.i.d. random rotations. Then
  \begin{align}\label{eq:9}
    \Prob\left\{\eqref{eq:erc}\text{ holds} 
    \right\}
    &= \Prob\left\{ -\mtx Q_i D_i \cap \sumnl_{j\ne i} \mtx Q_j D_j = \{\zerovct\}\text{ for each } i = 1,\dotsc,n,n+1 \right\}.
\intertext{  Under the same conditions,}
    \Prob\left\{ \eqref{eq:src}\text{ holds}
    \right\} 
    &= 
    \Prob\left\{ \left\llangle-\mtx Q_i
        D_i , \sumnl_{j\ne i} \mtx Q_j D_j\right\rrangle < 1
      \text{ for each } i = 1,\dotsc,n,n+1\right\}.\label{eq:10}
  \end{align}
\end{lemma}
\noindent
The proof, which appears in Appendix~\ref{sec:rand-nullsp-proof}, requires only an elementary application of the rotation invariance of the Haar measure.

\subsubsection{Exchanging stable for exact recovery}
\label{sec:stab-almost-equiv}

Our second simplification shows that the stability condition~\eqref{eq:src}  holds with the same probability that the recovery condition~\eqref{eq:erc} holds.  
\begin{lemma}\label{lem:stab-exact-equiv}
  The probabilities appearing in~\eqref{eq:9} and~\eqref{eq:10} are equal.
\end{lemma}
\noindent This result is immediate for closed cones: compactness arguments imply that two closed cones do not intersect if and only if the angle between the cones is strictly less than one.  Hence,~\eqref{eq:erc} is equivalent to~\eqref{eq:src} when all of the descent cones \(D_i\)  are closed.   The proof of Lemma~\ref{lem:stab-exact-equiv} in Appendix~\ref{sec:equiv-betw-stab} shows that this equivalence almost surely holds even when the cones are not closed.

\subsubsection{Polarizing the exact recovery condition}
\label{sec:polar-exact-recov-2}
Our final simplification reduces the  \(n+1\) intersections in~\eqref{eq:9} to a single intersection. 
\begin{lemma}\label{lem:polar-exact}
  Suppose that \(D_i \ne \{\zerovct\}\) for at least two indices \(i\in \{1,\dotsc,n,n+1\}\).  Then
  \begin{multline}
    \label{eq:polar-exact}
    \Prob\left\{ -\mtx Q_i D_i \cap \sumnl_{j\ne i} \mtx Q_j D_j = \{\zerovct\}\text{ for each } i = 1,\dotsc,n,n+1 \right\} \\ =\Prob\bigl\{\mtx Q_1 D_1^\polar \cap\dotsb\cap \mtx Q_{n} D_{n}^\polar \cap \mtx Q_{n+1} D_{n+1}^\polar \ne \{\zerovct\}\bigr\}
  \end{multline}
  where the matrices \((\mtx Q_i)_{i=1}^n\) are drawn i.i.d.\ from the random orientation model.
\end{lemma}
\noindent  The demonstration appears in Appendix~\ref{sec:polar-exact-recov}, but we describe main difficulty here.   Let \(C,D \subset \R^d\) be two cones such that \(-C\cap D = \{\zerovct\}\).  The separating hyperplane theorem provides  a nonzero \(\vct w \in \R^d\) that weakly separates \(-C\) and \(D\):
\begin{equation*}
  \langle \vct w,-\vct x\rangle \ge 0 \text{ for all } \vct x \in C \qtq{and} \langle \vct w,\vct y\rangle \le 0 \text{ for all } \vct y \in D.
\end{equation*}
By definition of polar cones, we have \(\vct w \in C^\polar\cap D^\polar\), so that polar cones intersect nontrivially. 

On the other hand, reversing the argument above shows that any nonzero  \( \vct w \in C^\polar \cap D^\polar \ne \{\zerovct\}\) weakly separates \(-C\) from \(D\).  Unfortunately, weak separation is not enough to conclude the strong separation \(-C\cap D = \{\zerovct\}\).  Proposition~\ref{prop:polar-prop} in Appendix~\ref{sec:polar-exact-recov} shows that the event \(C^\polar\cap D^\polar \ne \{\zerovct\}\) {almost surely} implies the event \(-C\cap D = \{\zerovct\}\)  when \(C\) and \(D\) are randomly oriented.  The proof of Lemma~\ref{lem:polar-exact} bootstraps this result to the multiple cone case.

\subsection{The approximate kinematic formula}
\label{sec:appr-kinem-form}

The simplifications in Section~\ref{sec:two-simplifications-} reduce the  study of~\eqref{eq:erc} and~\eqref{eq:src} to the question of computing the probability~\eqref{eq:polar-exact} that randomly oriented cones intersect.  Remarkably, formulas for the probability that two randomly oriented cones share a ray appear in literature on stochastic geometry under the name \emph{kinematic formulas}~\cite{San:76,Gla:95}. While exact, these formulas involve  geometric parameters that are typically difficult to compute. 

In recent work, the present authors and collaborators demonstrate that the classical kinematic formulas can be summarized using the statistical dimension~\cite[Thm.~7.1]{AmeLotMcC:13}. The following result extends this formula to the intersection of an arbitrary number of randomly oriented cones.

\begin{theorem}[Approximate kinematic formula]\label{thm:appr-kin}
  Let \(C_1,\dotsc,C_{n-1},C_n\in \cC_d\) be closed, convex cones and \(L\subset \R^d\) an \(m\)-dimensional linear subspace. Define the parameters
  \begin{equation}\label{eq:omega-theta}
    \Omega := \sumnl_{i=1}^n \sdim(C_i) \qtq{and} \theta^2 := \sumnl_{i=1}^n \sdim(C_i)\wedge \sdim(C_i^\polar) .
  \end{equation}
  Suppose that \((\mtx Q_i)_{i=1}^{n+1}\) are i.i.d. random rotations.  Then for any \(\lambda >0\), 
  \begin{align}\label{eq:appr-kin-upper}
   \Omega + m \le nd-\lambda  &\implies \Prob\bigl\{\mtx Q_1 C_1\cap \dotsb\cap \mtx Q_n C_n \cap \mtx Q_{n+1} L \ne \{\zerovct\}\bigr\} \le \phantom{1{}-{}} p_\theta(\lambda); \\
    \label{eq:appr-kin-lower}
\Omega+ m \ge nd+\lambda  &\implies \Prob\bigl\{\mtx Q_1 C_1 \cap \dotsb\cap \mtx Q_n C_n \cap \mtx Q_{n+1} L \ne \{\zerovct\}\bigr\} \ge 1-p_\theta(\lambda).
  \end{align}
  The concentration function \(p_\theta(\lambda)\) is defined for \(\lambda >0\) by
  \begin{equation}\label{eq:pclambda}
    p_\theta(\lambda) :=  \exp\left(\frac{-\lambda^2/4}{\theta^2+ \lambda/3}\right).
  \end{equation}
\end{theorem}
\noindent The proof of this result forms the topic of Appendix~\ref{sec:proof-appr-kinem}.  The argument requires some background from conic integral geometry that we provide in Appendix~\ref{sec:conic-integr-geom}. The proof of Theorem~\ref{thm:appr-kin} appears in Appendix~\ref{sec:proof-appr-kinem-1}. %

\subsection{Completing the proof}
\label{sec:completing-proof}

At this point, we have presented all of the  components needed to complete the proof of Theorem~\ref{thm:mult-bd}.  Let us summarize the progress. Lemma~\ref{lem:achievable-stab} shows that \eqref{eq:erc} and~\eqref{eq:src} characterize  exact and stable recovery.   Under the random orientation model, the probability that the stable recovery condition~\eqref{eq:src} holds is equal to the probability that the exact recovery condition~\eqref{eq:erc}  holds (Lemma~\ref{lem:stab-exact-equiv}).  We have also seen that
\begin{equation}\label{eq:erc-polar}
  \Prob\{\eqref{eq:erc}\text{ holds}\} = \Prob\bigl\{\mtx Q_1 D_1^\polar \cap\dotsb\cap \mtx Q_{n} D_{n}^\polar \cap \mtx Q_{n+1} D_{n+1}^\polar \ne \{\zerovct\}\bigr\}
\end{equation}
so long as \(D_i\ne \{\zerovct\}\) for at least two indices \(i\in \{1,\dotsc,n,n+1\}\) (Lemmas~\ref{lem:rand-nullspace} and~\ref{lem:polar-exact}).  

To complete the proof of Theorem~\ref{thm:mult-bd},  we use the approximate kinematic formula of Theorem~\ref{thm:appr-kin} to develop  lower and upper bounds on the probability~\eqref{eq:erc-polar}. This establishes the implications~\eqref{eq:success-mult-demix} and~\eqref{eq:fail-mult-demix}  when \(D_i\ne \{\zerovct\}\) for at least two indices \(i\).  We defer the  degenerate case where  \(D_i=\{\zerovct\}\) for all except (possibly) one index \(i\) to Appendix~\ref{sec:degenerate-case}.

\begin{proof}[Proof of Theorem~\ref{thm:mult-bd}] 

  We assume that \(D_i\ne \{\zerovct\}\) for at least two indices \(i\).   The polarity formula for the statistical dimension~\eqref{eq:sdim-polar} implies
  \begin{equation*}%
    \sumnl_{i=1}^n \sdim(D_i^\polar)+ m = \sumnl_{i=1}^n \left(d- \sdim(\overline{D}_i)\right) +m = nd - \Delta+m,
  \end{equation*}
  where we use the fact that \(\delta_i=\sdim(\overline{D}_i)\) by definitions~\eqref{eq:sdim_i} and~\eqref{eq:cone-defs} of \(\delta_i\) and \(D_i\).  For the same reason, we have
  \begin{equation*}
    \sigma^2  = \sumnl_{i=1}^n \sdim(D_i^\polar)\wedge \sdim(\overline{D}_i).
  \end{equation*}
  where  the width parameter \(\sigma\) is defined in~\eqref{eq:Delta-sigma}. Moreover,  definition~\eqref{eq:cone-defs} shows that the cone \(D_{n+1}^\polar\) is a linear subspace with
  \begin{equation*}
    \dim(D_{n+1}^\polar)=\dim(\nullity(\mtx A)^\perp)=  m
  \end{equation*}
  because \(\mtx A\in \R^{m\times d}\) has full row rank by assumption.  

  Therefore, when \(m \ge \Delta+ \lambda_*\) the lower bound~\eqref{eq:appr-kin-lower} of the approximate kinematic formula implies %
  \begin{equation}\label{eq:lower-pf}
    \Prob\bigl\{\mtx Q_1 D_1^\polar \cap\dotsb\cap \mtx Q_{n} D_{n}^\polar \cap \mtx Q_{n+1} D_{n+1}^\polar \ne \{\zerovct\}\bigr\}  \ge 1-p_\sigma(\lambda_*). %
  \end{equation}  
  Similarly, when \(m \le \Delta - \lambda_*\), the upper bound~\eqref{eq:appr-kin-upper} of the approximate kinematic formula provides
  \begin{equation}\label{eq:upper-pf}
    \Prob\bigl\{\mtx Q_1 D_1^\polar \cap\dotsb\cap \mtx Q_{n} D_{n}^\polar \cap \mtx Q_{n+1} D_{n+1}^\polar \ne \{\zerovct\}\bigr\}  \le p_\sigma(\lambda_*)%
  \end{equation}
  
In light of Lemmas~\ref{lem:achievable-stab}, \ref{lem:rand-nullspace}, \ref{lem:stab-exact-equiv}, and~\ref{lem:polar-exact},  the inequalities~\eqref{eq:lower-pf} and~\eqref{eq:upper-pf} imply claims~\eqref{eq:success-mult-demix} and~\eqref{eq:fail-mult-demix}  once we verify that
  \begin{equation}\label{eq:psig-bd}
    p_\sigma(\lambda_*) \le \eta.
  \end{equation}
  To verify~\eqref{eq:psig-bd}, we  invert the definition~\eqref{eq:pclambda} of \(p_\sigma\) and solve a quadratic equation to find
\begin{equation*}
  p_\sigma(\lambda) \le \eta \iff \lambda \ge \frac{2}{3}\left(L + \sqrt{L^2 + 9 L \sigma^2}\right),
\end{equation*}
where \(L:=\log(1/\eta)\).  Since \(\sqrt{a^2+b^2} \le a+b\) for positive \(a\) and \(b\), we see
\begin{equation*}
  \frac{2}{3}\left(L + \sqrt{L^2 + 9 L \sigma^2}\right) \le \frac{4}{3} L + 2 \sigma\sqrt{L} =: \lambda_*.
\end{equation*}
Thus~\eqref{eq:psig-bd} holds for our choice \(\lambda_*\), as claimed.  This completes the proof in the case where  \(D_i \ne  \{\zerovct\}\) for at least two indices \(i\).  We complete the proof for the remaining case in Appendix~\ref{sec:degenerate-case}.
\end{proof}

\section{Numerical examples}
\label{sec:applications}

In this section, we describe two simple numerical experiments that demonstrate the accuracy of Theorem~\ref{thm:mult-bd}. In our first example, we consider demixing three components, two of them sparse, the third a sign vector.  Our second example considers demixing two sparse vectors with undersampling. Technical details about the experiments are collected in Appendix~\ref{sec:numerical-details}.
\begin{figure}[t!]
  \centering
  \includegraphics[width=0.5\columnwidth]{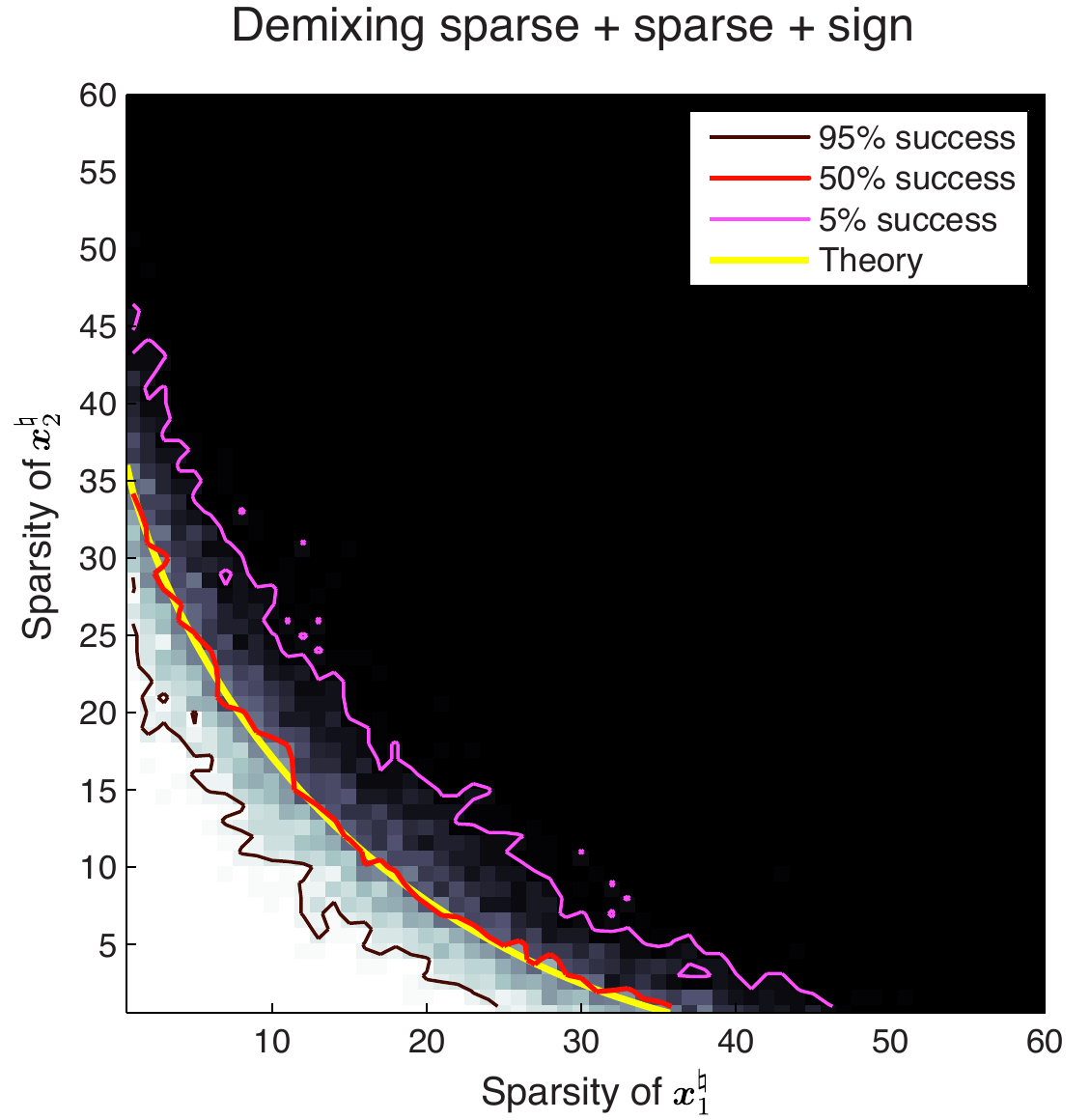}      
  \includegraphics[width=0.48\columnwidth]{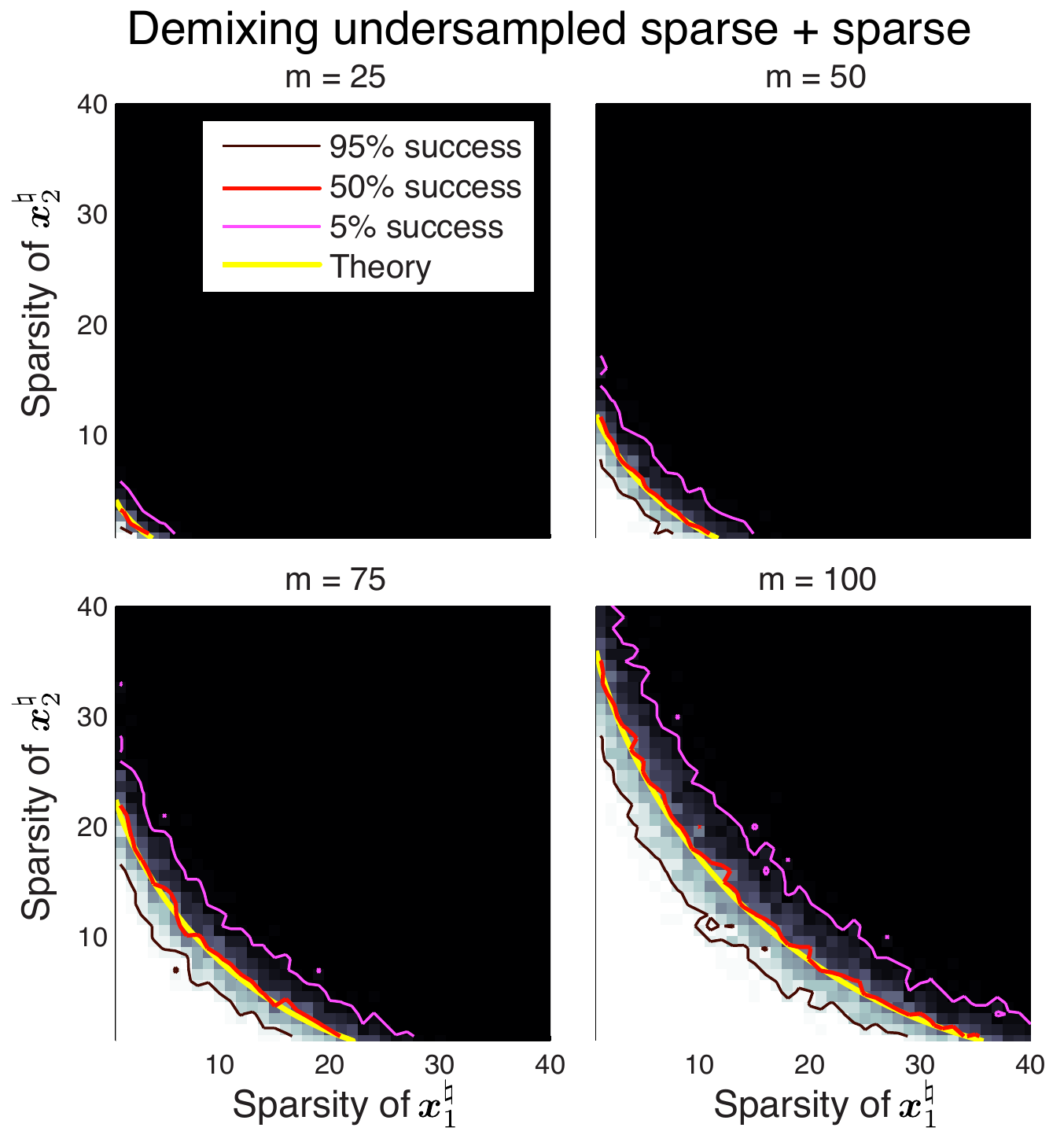}
  \caption{\textsf{Demixing experiments from Section~\ref{sec:applications}.}  The colormaps display the empirical probability of successful demixing, from complete success \textsl{(white)} to total failure \textsl{(black)}.  The transition region \textsl{(gray)} contains a mixture of successes and failures. Three contour lines indicate \(95\%\) \textsl{(brown)}, \(50\%\) \textsl{(red)}, and \(5\%\) \textsl{(pink)} empirical success lines.  The yellow curve indicates where an  approximation to the total statistical dimension \(\Delta\) equals to the number of measurements \(m\).  \textbf{[Left]} Demixing two sparse vectors from a sign vector in dimension \(d=200\) with complete measurements.  \textbf{[Right]} Demixing two sparse vectors in dimension \(d=200\) from \(m=25,50,75,100\) measurements.   }
  \label{fig:sparse-sparse-sign}
\end{figure}

\paragraph*{Sparse, sparse,  and sign}
\label{sec:sparse-sparse-sign}

In  our first  experiment, we fix the ambient dimension \(d=200\) and generate a mixed observation of the form
\begin{equation*}
  \vct z_0 = \mtx U_1 \vct x_1^\nat + \mtx U_2 \vct x_2^\nat + \mtx U_3 \vct x_3^\nat \in \R^d,
\end{equation*}
where \(\vct x_1^\nat\) and \(\vct x_2^\nat\) are sparse vectors, \(\vct x_3^\nat\in \{\pm 1\}^d\) is a sign vector, and the tuple \((\mtx U_i)_{i=1}^3\) consists of i.i.d. random rotations. In order to demix this observation, we solve the constrained demixing program
\begin{equation}\label{eq:sss-demix}
  \begin{aligned}
    \minimizeOp_{\vct x_i \in \R^d}\quad &\enormsq{ \mtx U_1\vct x_1+\mtx U_2\vct x_2+\mtx U_3\vct x_3 - \vct z_0 }  
    \\ 
    \subjectto &  \lone{\vct x_1}\le \lone{\vct x_1^\nat}, \quad \lone{\vct x_2}\le \lone{\vct x_2^\nat}, \qtq{and} \linf{\vct x_3}\le \linf{\vct x_3^\nat},
  \end{aligned}
\end{equation}
where \(\linf{\vct x}:= \max_{i=1,\dotsc,d} |x_i|\) is the \(\ell_\infty\) norm that is a convex penalty function associated to the binary sign vectors \(\{\pm 1\}^d\).   

Figure~\ref{fig:sparse-sparse-sign} [left] shows the results of this experiment as the sparsity of \(\vct x_1^\nat\) and \(\vct x_2^\nat\) vary.  The colormap indicates the empirical probability of success over \(35\) trials.  The yellow curve uses provably accurate formulas from~\cite[Sec.~4]{AmeLotMcC:13} to approximate the location where
\begin{equation*}
  \sdim(\lone{\cdot},\vct x_1^\nat) +  \sdim(\lone{\cdot},\vct x_2^\nat) +   \sdim(\linf{\cdot},\vct x_3^\nat) = d.
\end{equation*}
The agreement between the \(50\%\) empirical success curve and the theoretical yellow curve is remarkable.  See Appendix~\ref{sec:numerical-details} for further details.

\paragraph*{Undersampled sparse and sparse}
\label{sec:unders-sparse-sparse}

In our second experiment, we  fix the ambient dimension \(d=200\) and consider demixing the observation
\begin{equation*}
  \vct z_0 = \mtx A\bigl(\mtx U_1\vct x_1^\nat + \mtx U_2\vct x_2^\nat\bigr) ,
\end{equation*}
where \(\mtx A \in \R^{m\times d}\) has full row rank, the constituents \(\vct x_1^\nat\) and \(\vct x_2^\nat\) are sparse, and \(\mtx U_1\) and \(\mtx U_2\) are drawn from the random orientation model.  We demix the observation by solving
\begin{equation}\label{eq:ssu}
  \begin{aligned}
    \minimizeOp_{\vct x_i \in \R^d}\quad &\enormsq{ \mtx A^\psinv\bigl( \mtx A(\mtx U_1\vct x_1 + \mtx U_2 \vct x_2)- \vct z_0 \bigr)}  
    \\ 
    \subjectto &  \lone{\vct x_1}\le \lone{\vct x_1^\nat} \qtq{and} \lone{\vct x_2}\le \lone{\vct x_2^\nat}
  \end{aligned}
  \end{equation}
  The results of this experiment with \(m=25,50,75\) and \(100\) appear in Figure~\ref{fig:sparse-sparse-sign} [right].  The colormaps indicate the empirical probability of success over \(35\) trials as the sparsity of \(\vct x_1^\nat\) and \(\vct x_2^\nat\) varies. The yellow curve approximates the location where
  \begin{equation*}
    \sdim(\lone{\cdot},\vct x_1^\nat)+    \sdim(\lone{\cdot},\vct x_2^\nat) = m.
  \end{equation*}
  Once again, this yellow curve agrees very well with the red \(50\%\) empirical success contour. 

\section{Conclusions and open problems}
\label{sec:open-problems}

This work unifies and resolves a number of theoretical questions regarding when it is possible to demix a superposition of incoherent signals.  Under our random incoherence model, we find  that demixing is possible if and only if the total number of measurements is greater than the total statistical dimension.  While this result provides an intuitive and unifying theory for a large class of demixing problems, there are several important open problems that must be addressed before a complete ``theory of demixing'' emerges.
\begin{description}
\item[Lagrange parameters] Most of the prior work on demixing provides guarantees under explicit choices of the Lagrange parameter for~\eqref{eq:noisy-demix-lag}, yet to the best of our knowledge, the only work that demonstrates sharp recovery bounds with specified Lagrange parameters occur for the LASSO problem, where \(n=1\) and \(f_1=\lone{\cdot}\)~\cite{BayMon:12}.  Very recent work of Stojnic~\cite{Sto:13} achieves comparable guarantees using a different approach.

Explicit choices of Lagrange parameters appear in~\cite{FoyMac:13}. These choices provide near-optimal empirical performance, but currently there is no proof that their choice of parameters reaches the phase transition that we identify.  It would be very interesting to provide provably optimal choices of Lagrange parameters for demixing.

\item[Other random models] Our numerical experience indicates that the incoherence model considered in this work is predictive for highly incoherent situations.  However, these results appear overly optimistic in more coherent situations.  The difference between these situations  appears, for example, in an application to calcium imaging~\cite[Fig.~3]{PnePan:13}.  Extending our results to other incoherence models will clarify where the phase transition in Theorem~\ref{thm:mult-bd}  predicts empirical performance, and where it does not.   

\item[Statistical dimension calculations] For practical applications of this work, we require  accurate statistical dimension calculations. A recipe for these computations put forward in~\cite{ChaRecPar:12} has provable guarantees under some technical conditions (cf.~\cite[Sec.~4.4]{AmeLotMcC:13} and~\cite[Prop.~1]{FoyMac:13}), but expressions for the statistical dimension of a number of important convex regularizers remains unknown. New statistical dimension computations immediately extend the reach of the methods used in this paper. 
\end{description}

\appendix

\section{Deterministic conditions}
\label{sec:determ-cond}

This section provides the deterministic demixing claims of Lemma~\ref{lem:achievable-stab}. The exact recovery conditions for the noiseless setting appear in Section~\ref{sec:exact-recovery}, and the stable recovery guarantees appear in Section~\ref{sec:stable-recovery}.

\subsection{Exact recovery}
\label{sec:exact-recovery}
 In this section, we show that, in the noiseless setting \(\vct w = \zerovct\) , the tuple \((\vct x_i^\nat)_{i=1}^n\) is the unique optimal point of~\eqref{eq:const-dmx} if and only if~\eqref{eq:erc} holds.  

 Suppose first that~\eqref{eq:erc} holds,  and let \((\hvct x_i)_{i=1}^n\) be any optimal point of~\eqref{eq:const-dmx}. Define the vectors
  \begin{equation}\label{eq:yi-def}
    \vct y_i := \hvct x_i - \vct x_i^\nat \qtq{for}i=1,\dotsc,n-1,n \qtq{and} 
    \vct y_{n+1}:= -\sumnl_{i=1}^n \mtx U_i (\hvct x_i -\vct x_i^\nat). 
  \end{equation}
  We  will show that \(\vct y_{i}=\zerovct\) for each \(i=1,\dotsc,n,n+1\).    The tuple \((\vct x_i^\nat)_{i=1}^n\) is trivially feasible for~\eqref{eq:const-dmx}, and the objective at this  point is given by
  \begin{equation}\label{eq:6}
    \enorm{\mtx A^\psinv \left(\mtx A \sumnl_{i=1}^n \mtx U_i\vct x_i^\nat - \vct  z_0 \right)} = \enorm{\mtx A^\psinv \mtx A\zerovct } =    0,
  \end{equation}
 by definition~\eqref{eq:obs-model} of \(\vct z_0\) and the assumption \(\vct w=\zerovct\).  Since \((\hvct x_i)_{i=1}^n\) is optimal for~\eqref{eq:cone-defs} by assumption, its objective value must be less than the value given in~\eqref{eq:6}, which implies
  \begin{equation*}
0 \ge \enorm{\mtx A^\psinv\left( \mtx A \sumnl_{i=1}^n \mtx U_i\hvct x_i -\vct z_0\right)} =   \enorm{\mtx A^\psinv\mtx A (-\vct y_{n+1})} \ge 0
  \end{equation*}
  where the equality follows by the definition~\eqref{eq:yi-def}  of \(\vct y_{n+1}\).  These sandwiched inequalities imply that 
  \begin{equation}\label{eq:1}
    \vct y_{n+1} \in \nullity(\mtx A)= D_{n+1}.
  \end{equation}
  by the definition~\eqref{eq:cone-defs}  of  \(D_{n+1}\).  Since an optimal point of~\eqref{eq:const-dmx} is also feasible for~\eqref{eq:const-dmx}, we also have \( f_i(\hvct x_i)\le  f_i(\vct x_i^\nat)\).  The definition~\eqref{eq:desc-cone-def} of the descent cone implies that
  \begin{equation}\label{eq:2}
    \vct y_i =  (\hvct x_i -\vct x_i^\nat)\in \Desc(f_i,\vct x_i^\nat)= D_i.
  \end{equation}
  By expanding the definition of \(\vct y_i\), we find the trivial relation
  \begin{equation*}
    \sumnl_{i=1}^{n+1}\vct U_{i} \vct y_i = \sumnl_{i=1}^n \mtx U_i(\hvct x_i -\vct x_i^\nat)-\mtx U_{n+1}\left( \sumnl_{i=1}^n \mtx U_i(\hvct x_i -\vct x_i^\nat)\right) = \zerovct,
  \end{equation*}
where we recall that \(\mtx U_{n+1}=\Id\) by definition~\eqref{eq:cone-defs}. Upon rearrangement, this equation is equivalent to
  \begin{equation*}
    \mtx U_i \vct y_i = - \sumnl_{j\ne i}\mtx U_j \vct y_j \qtq{for each} i =1,\dotsc,n,n+1.
  \end{equation*}
 Combined with the containments~\eqref{eq:1} and~\eqref{eq:2}, the relations above imply
 \begin{equation*}
   \mtx U_i \vct y_i \in \mtx U_i D_i \cap -\left(\sumnl_{j\ne i} \mtx U_j D_j\right) =\{\zerovct\}\qtq{for each} i = 1,\dotsc,n,n+1.
 \end{equation*}
 where the trivial intersection follows because the exact recovery condition~\eqref{eq:erc} is in force.  Since each \(\mtx U_i\) is invertible, we must have \(\vct y_i = \vct x_i^\nat-\hvct x_i  = \zerovct\) for each \(i=1,\dotsc,n\).  But \((\hvct x_i)_{i=1}^n\) was an arbitrary optimal point of~\eqref{eq:const-dmx},  so condition~\eqref{eq:erc} indeed implies that the tuple \((\vct x_i^\nat)_{i=1}^n\) is the unique optimum of~\eqref{eq:const-dmx}.

 Now suppose that~\eqref{eq:erc} does not hold.  Then there exists an index \(i_*\in \{1,\dotsc,n,n+1\}\) and a vector \(\vct y_{i_*}\ne \zerovct \) such that
 \begin{equation*}
   \mtx U_{i_*} \vct y_{i_*}  \in \mtx U_{i_*} D_{i_*} \cap -\left(\sumnl_{j\ne i_*}\mtx U_j D_j\right).
 \end{equation*}
  Equivalently,  there are vectors \(\vct y_j \in D_j\) such that
 \begin{equation}\label{eq:3}
   \mtx U_{i_*}\vct y_{i_*} = - \sumnl_{j\ne i_*} \mtx U_j \vct y_j. 
 \end{equation}
It follows from the definition of the descent cone and a basic convexity argument (cf.~\cite[Prop.~2.4]{McCTro:12})  that for some sufficiently small \(\tau >0\),
 \begin{equation}\label{eq:4}
    f_i(\tau \vct y_i + \vct x_i^\nat) \le f_i(\vct x_i^\nat) \qtq{for all} i = 1,\dotsc,n-1,n.
 \end{equation}
 Define \(\hvct x_i := \tau \vct y_i + \vct x_i^\nat\).    The definition \(\vct z_0 = \mtx A\left(\sum_{i=1}^n \mtx U_i \vct x_i^\nat\right)\) and~\eqref{eq:3} implies
 \begin{equation}\label{eq:5}
   \enorm{\mtx A^\psinv \left(\mtx A \sumnl_{i=1}^n \mtx U_i\hvct x_i - \vct  z_0 \right)} =     \enorm{\mtx A^\psinv \mtx A (-\tau\mtx U_{n+1}\vct y_{n+1})} = 0.
 \end{equation}
 The final equality follows because \(\mtx U_{n+1}\vct y_{n+1} \in \nullity(\mtx A)\) by definition~\eqref{eq:cone-defs}.

To summarize, Equation~\eqref{eq:4} shows that the tuple \((\hvct x_i)_{i=1}^n\) is feasible for the constrained demixing program~\eqref{eq:const-dmx}, while Equation~\eqref{eq:5} indicates that the objective value at \((\hvct x_i)_{i=1}^n\) is the minimum possible.  Therefore, \((\hvct x_i)_{i=1}^n\) is an optimal point of~\eqref{eq:const-dmx}. Since \(\vct y_{i_*}\ne \zerovct\)  and \(\tau>0\), we see  that \((\hvct x_i)_{i=1}^n \ne (\vct x_i^\nat)_{i=1}^n \).  We conclude that  \((\vct x_i^\nat)_{i=1}^n\) is \emph{not} the unique optimal point of~\eqref{eq:const-dmx} when~\eqref{eq:erc} fails to hold, which completes the exact recovery portion of Lemma~\ref{lem:achievable-stab}.

\subsection{Stable recovery}
\label{sec:stable-recovery}
The stable recovery claims of Lemma~\ref{lem:achievable-stab}  immediately follow from the next result.
\begin{lemma}[Stability of constrained demixing]\label{lem:const-noise}
  With the notation from definition~\eqref{eq:cone-defs}, suppose that there exists an \(\eps \in (0,1)\) such that
  \begin{equation}\label{eq:cone-separation}
    \left\llangle -\mtx U_i D_i, \sumnl_{j\ne i}\mtx U_j D_j\right\rrangle \le 1-\eps \qtq{for all} i = 1,\dotsc, n, n+1.
  \end{equation}
  Then the error bound
  \begin{equation}\label{eq:const-err-bd}
    \enormsq{\hvct x_i - \vct x_i^\nat} 
    \le \frac{1}{\eps} \enormsq{\vct w} \qtq{for all} i = 1,\dotsc,n-1,n
  \end{equation}
holds for any optimal point \((\hvct x_i)_{i=1}^n\) of~\eqref{eq:const-dmx}.
\end{lemma}
\begin{proof}[Proof of Lemma~\ref{lem:achievable-stab} from Lemma~\ref{lem:const-noise}]
  Whenever the stable recovery condition~\eqref{eq:src} holds, there is
  an \(\eps>0\) such that the condition~\eqref{eq:cone-separation} of
  Lemma~\ref{lem:const-noise} holds.  But~\eqref{eq:const-err-bd} is
  equivalent to the definition~\eqref{eq:stab-recov-defn} of stable
  recovery with constant \(c= \eps^{-1/2}\).  
\end{proof}
\noindent
The proof of Lemma~\ref{lem:const-noise} rests on the following elementary observation. 
\begin{proposition}\label{prop:basic-obs}
  Suppose \(\langle \vct x,\vct y\rangle \ge -(1-\eps) \enorm{\vct x}\enorm{\vct y}\) for some \(\eps \in (0,1]\).   Then
  \begin{equation*}
    \enormsq{\vct x} + \enormsq{\vct y}\le \frac{1}{\eps} \enormsq{\vct x+\vct y}.
  \end{equation*}
\end{proposition}
\begin{proof}
  We have the following string of inequalities:
  \begin{align*}
    \eps\left(\enormsq{\vct x} + \enormsq{\vct y}\right) 
    &\le 
    \eps\left(\enormsq{\vct x} + \enormsq{\vct y}\right)  + (1-\eps)\left(\enorm{\vct x}-\enorm{\vct y}\right)^2 
    \\ & =
    \enormsq{\vct x} + \enormsq{\vct y} - 2(1-\eps) \enorm{\vct x}\enorm{\vct y}
    \\ & \le
    \enormsq{\vct x}+\enormsq{\vct y} + 2\langle \vct x,\vct y\rangle  
  \end{align*}
  The first line follows because squares are nonnegative and \(\eps \le 1\), the second line is algebra, and the final expression relies on our assumption on \(\langle\vct x,\vct y\rangle\).  The last expression is \(\enormsq{\vct x+ \vct y}\).
\end{proof}

\begin{proof}[Proof of Lemma~\ref{lem:const-noise}]
  Define the vectors \(\tvct y_i := (\hvct x_i - \vct x_i^\nat)\) for \(i=1,\dotsc,n-1,n\), and
  \begin{equation}\label{eq:7}
    \tvct y_{n+1}:= ( \mtx A^\psinv \mtx A-\Id)\left(\sumnl_{i=1}^n \mtx U_i(\hvct x_i-\vct x_i^\nat)\right).
  \end{equation}
  Since \(f_i(\hvct x_i) \le  f_i(\vct x_i^\nat)\) for all \(i=1,\dotsc, n\), we have \(\tvct y_i \in \mtx \Desc(f_i,\vct x_i^\nat)\).  Moreover, the operator \((\Id - \mtx A^\psinv\mtx A)\) is the projection onto the nullspace of \(\mtx A\), so that \(\tvct y_{n+1} \in\nullity(\mtx A)\).  Applying definition~\eqref{eq:cone-defs} of the cones \(D_i\), we see
  \begin{equation*}
    \tvct y_i \in  D_i \qtq{for every}  i =1,\dotsc,n,n+1.
  \end{equation*}

  By assumption,  \(\llangle-\mtx U_i D_i,\sumnl_{j\ne i}\mtx U_j D_j\rrangle\le 1-\eps\) for each \(i=1,\dotsc,n,n+1\), so that
  \begin{equation*}
 \left   \langle -\mtx U_i\tvct y_i ,\sumnl_{j\ne i}\mtx U_j\tvct y_i\right \rangle \ge -(1-\eps) \enorm{ \mtx U_i\tvct y_i}\enorm{\sumnl_{j\ne i}\mtx U_j\tvct y_j}
  \end{equation*}
  by definition~\eqref{eq:conic-dist} of the angle between cones. Proposition~\ref{prop:basic-obs} provides the inequality
  \begin{equation}\label{eq:8}
    \enormsq{\tvct y_i}\le     \enormsq{\mtx U_i\tvct y_i}+ \enormsq{\sumnl_{j\ne i}\mtx U_j\tvct y_j} \le \frac{1}{\eps} \enormsq{\sumnl_{j=1}^{n+1} \mtx U_j \tvct y_j},
  \end{equation}
  where we use the fact that \(\enorm{\tvct y_i}=\enorm{\mtx U_i\tvct y_i}\) because \(\mtx U_i\) is orthogonal.  Expanding the definitions of \(\tvct y_i\) and \(\vct z_0\), we calculate
  \begin{align*}
    \enormsq{\sumnl_{i=1}^{n+1} \mtx U_i \tvct y_i} &= \enormsq{\mtx A^\psinv\left( \mtx A\sumnl_{i=1}^n \hvct x_i - \vct z_0\right) +(\Id -\mtx A^\psinv \mtx A) \vct w   }    
    \\ 
    &= 
    \enormsq{\mtx A^\psinv\left( \mtx A\sumnl_{i=1}^n \hvct x_i - \vct z_0\right)}
    +
    \enormsq{(\Id -\mtx A^\psinv \mtx A) \vct w   }    
    \\
    &\le 
    \enormsq{\mtx A^\psinv\left( \mtx A\sumnl_{i=1}^n \vct x_i^\nat - \vct z_0\right)  } +    \enormsq{(\Id -\mtx A^\psinv \mtx A) \vct w   }     
    \\ 
        &=    \enormsq{\mtx A^\psinv \mtx A \vct w} +\enormsq{(\Id-\mtx A^\psinv \mtx A) \vct w} .
  \end{align*}
  The second equality holds because \((\Id - \mtx A^\psinv \mtx A)\mtx A^\psinv= \zerovct\).  For the inequality,  note that  \((\hvct x_i)_{i=1}^n\) minimizes~\eqref{eq:const-dmx} and the tuple \((\vct x_i^\nat)_{i=1}^n\) is feasible for~\eqref{eq:const-dmx}.  The final equality is the definition~\eqref{eq:obs-model} of \(\vct z_0\). Combining the bound above with~\eqref{eq:8}, we see
  \begin{equation*}
    \enormsq{\hvct x_i-\vct x_i^\nat} = \enormsq{\hvct y_i}  \le \frac{1}{\eps} \left(  \enormsq{\mtx A^\psinv \mtx A \vct w}+\enormsq{(\Id - \mtx A^\psinv \mtx A )\vct w} \right) = \frac{1}{\eps} \enormsq{\mtx w},
  \end{equation*}
  where the final relation follows by orthogonality. 
\end{proof}

\section{Simplifying results}

This section presents the proofs of the lemmas appearing in Section~\ref{sec:two-simplifications-}. 

\subsection{Randomizing the nullspace}\label{sec:rand-nullsp-proof}

 Lemma~\ref{lem:rand-nullspace} is an easy consequence of a basic fact about invariant measures.
\begin{fact}
  Let \((\mtx Q_1,\dotsc, \mtx Q_{n-1}, \mtx Q_n)\) be i.i.d.\ random rotations in \(\orth{d}\).  Suppose that \(f\colon (\orth{d})^n\to \R\) is a measurable function that satisfies
  \begin{equation}\label{eq:forget-condition}
    \Expect\Bigl[ \Expect\left[\abs{ f(\mtx Q_1,\dotsc,\mtx Q_{n-1},\mtx Q_n)} \;\big\vert\; \mtx Q_1 \right]\Bigr]  < \infty,
  \end{equation}
  where the outer expectation is over \(\mtx Q_1\), and the inner expectation is over \(\mtx Q_i\) for \(i\ge 2\).  In particular, condition~\eqref{eq:forget-condition} holds when \(\abs{f}\) is bounded. Then
  \begin{equation}\label{eq:forgetting}
    \Expect[f(\mtx Q_1,\mtx Q_2,\dotsc,\mtx Q_{n-1},\mtx Q_n)] 
    = 
    \Expect[f(\mtx Q_1,\mtx Q_1 \mtx Q_2,\dotsc,\mtx Q_{1} \mtx Q_{n-1},\mtx Q_1\mtx Q_n)].
  \end{equation}
 \end{fact}
 \noindent  The elementary proof is a simple application of Fubini's theorem and the definition of an invariant measure.  See~\cite[Fact~3.1]{McC:13} for a detailed proof.
 \begin{proof}[Proof of Lemma~\ref{lem:rand-nullspace}]
   Let \(\pInd_E \colon (\orth{d})^{n+1} \to \R\)  be the indicator function on the event
   \begin{equation*}
     E:=\left\{(\mtx U_i)_{i=1}^{n+1}\mid -\mtx U_i D_i \cap \sumnl_{j\ne i} \mtx U_j D_j = \{\zerovct\} \text{ for each } i = 1,\dotsc,n,n+1\right\},
   \end{equation*}
   where we recall that  the indicator function \(\pInd_S\) on a set \(S\)  is given by
   \begin{equation*}
     \pInd_S(x) :=
     \begin{cases}
       1, & x\in S, \\
       0, & \text{otherwise}.
     \end{cases}
   \end{equation*}
   Let \(\tilde{\mtx U}\in \orth{d}\) be a Haar distributed rotation independent of  \((\mtx U_i)_{i=1}^n\). Since \(\tilde{\mtx U}\{\zerovct\} = \{\zerovct\}\) for every rotation, we have the equality
   \begin{equation*}
     \pInd_E(\mtx U_1,\dotsc,\mtx U_n,\mtx U_{n+1}) = \pInd_{E}(\tilde{\mtx U}\mtx U_{1},\dotsc, \tilde{\mtx U}\mtx U_{n},\tilde{\mtx U}),
   \end{equation*}
   where we used the fact that \(\mtx U_{n+1} =\Id\). Taking expectations, we find
   \begin{align*}
     \Expect[\pInd_E(\mtx U_1,\dotsc,\mtx U_n,\mtx U_{n+1})] = 
     \Expect[\pInd_E(\tilde{\mtx U}\mtx U_{1},\dotsc, \tilde{\mtx U}\mtx U_{n},\tilde{\mtx U})]
     =
     \Expect[\pInd_E(\mtx Q_1,\dotsc,\mtx Q_{n},\mtx Q_{n+1})],
   \end{align*}
   where we arrive at the last line using~\eqref{eq:forgetting}.   The first claim~\eqref{eq:9} follows because the average value of the indicator function on an event is equal to the probability of that event.  The second claim~\eqref{eq:10} follows in a completely analogous manner, so we omit the details. 
 \end{proof}

\subsection{Equivalence between stability and exact recovery}
\label{sec:equiv-betw-stab}

The results below are  corollaries of an intuitive fact regarding the configuration of random cones.  We first need a definition.
\begin{definition}\label{def:touch}
Two cones \(C,D \in \cC_d\) are said to \term{touch} if they share a ray but are weakly separable by a hyperplane.
\end{definition}
\noindent When the cones are randomly oriented, touching is almost impossible. 
\begin{fact}[\protect{\cite[pp.~258--260]{SchWei:08}}]
  \label{fact:no-touching} Let \(C,D\in \cC_d\) be closed, convex
  cones such that both \(C, D\ne \{\zerovct\}\). Then
  \begin{equation*}
    \Prob\{\mtx Q C \text{ touches }   D \} = 0,
  \end{equation*}
  where \(\mtx Q\) is a random rotation in \(\orth{d}\).
\end{fact}

 We will also make use of the separating hyperplane theorem for convex cones due, in a much more general form,  to Klee~\cite[Thm.~2.5]{Kle:55}. 
\begin{fact}[Separating hyperplane theorem for convex cones]\label{fact:separation}
  Suppose \(C,C'\) are two convex cones in  \(\R^d\). If \(C\cap C' = \{\zerovct\}\), then there exists a nonzero \(\vct z \in\R^d \) such that \(\vct z \in C^\polar\) and \(-\vct z \in (C')^\polar\).  
\end{fact}
\begin{proof}[Proof of Lemma~\ref{lem:stab-exact-equiv}]
  The lemma claims that the events
  \begin{align}
   & \left\{ (\mtx Q_i)_{i=1}^{n+1}\mid -\mtx Q_i D_i \cap \sumnl_{j\ne i} \mtx Q_j D_j = \{\zerovct\}\text{ for each } i = 1,\dotsc,n,n+1 \right\} \qtq{and} \label{eq:9a}\\
   &\left\{(\mtx Q_i)_{i=1}^{n+1}\mid \left\llangle-\mtx Q_i
        D_i , \sumnl_{j\ne i} \mtx Q_j D_j\right\rrangle < 1
      \text{ for each } i = 1,\dotsc,n,n+1\right\}\label{eq:10a}
  \end{align}
  have equal probability when the matrices \(\mtx Q_i\) are drawn i.i.d.\ from the Haar measure on the orthogonal group~\(\orth{d}\).  The event appearing in~\eqref{eq:9a}  is  implied by the event appearing in~\eqref{eq:10a}, so the probability~\eqref{eq:9a} is larger than the probability~\eqref{eq:10a}.  We now show the reverse inequality.
  
  Fix any tuple \((\mtx Q_i)_{i=1}^{n+1}\) such that the  event~\eqref{eq:9a} holds, but that  the event~\eqref{eq:10a} \emph{does not} hold.  We claim that such a tuple must bring  two cones, out of a finite set, into touching position (Definition~\ref{def:touch}). The set of all such tuples must have probability zero by Fact~\ref{fact:no-touching} and the countable subadditivity of probability measures.  We conclude that the probability of~\eqref{eq:9a} is not larger than the probability of~\eqref{eq:10a}.
  
  We now establish the touching claim.  When the event in~\eqref{eq:10} does not hold,  there is an index \(i\) such that \(\bigl\llangle-\mtx Q_{i} D_{i},\sumnl_{j\ne i}\mtx Q_j D_j\bigr\rrangle = 1\).  By definition of the angle between cones, we  have 
\begin{equation}\label{eq:11}
  -\overline{\mtx Q_i D_i}\cap \overline{\sumnl_{j\ne i}\mtx Q_j
    D_j}\ne \{\zerovct\}.
\end{equation}
But because the event  in~\eqref{eq:9a} also holds, we have
\begin{equation*}
  -\mtx Q_i D_i \cap \sumnl_{j\ne i} \mtx Q_j D_j = \{\zerovct\}.
\end{equation*}
Hence the separating hyperplane theorem (Fact~\ref{fact:separation}) shows that \(-\mtx Q_iD_i\) and \(\sumnl_{j\ne i} \mtx Q_jD_j\) are weakly separable. By Definition~\ref{def:touch}, we see that the cones \(-\overline{\mtx Q_i  D_i} \) and \(\overline{\sumnl_{j\ne i} \mtx Q_j D_j}\) touch, as claimed.
\end{proof}

\subsection{Polarizing exact recovery}
\label{sec:polar-exact-recov}

We bootstrap the proof of Lemma~\ref{lem:polar-exact} from the analogous result for two cones.
\begin{proposition}\label{prop:polar-prop}
  Let \(C,D\subset \R^d\) be convex cones that contain zero.  If both \(C,D \ne \{\zerovct\}\), then the sets
  \begin{equation*}
    \Bigl\{\mtx Q \in \orth{d}\mid -C \cap \mtx Q D = \{\zerovct\}\Bigr\}
    \qtq{and}
    \Bigl\{\mtx Q \in \orth{d}\mid C^\polar \cap \mtx Q D^\polar\ne \{\zerovct\}\Bigr\}
  \end{equation*}
  coincide except on a set of Haar measure zero on \(\orth{d}\).
\end{proposition}
\begin{proof}
Suppose that both \(C,D\) are convex cones such that \(C,D\ne \{\zerovct\}\).   Whenever \(\mtx Q \in \orth{d}\) is such that \(-C\cap \mtx Q D = \{\zerovct\} \), the separating hyperplane theorem for convex cones (Fact~\ref{fact:separation}) ensures there exists a nonzero vector \(\vct w\in \R^d\) such that
  \begin{equation}\label{eq:sep1}
    \langle \vct w,\vct x \rangle \le 0 \; \text{for all } \vct x \in C
    \qtq{and} 
    \langle \vct w, \mtx Q\vct y \rangle \le 0 \; \text{for all } \vct y \in D.
  \end{equation}
  This is equivalent to the statement \(\vct w \in C^\polar \cap \mtx Q D^\polar\) by definition~\eqref{eq:polar-cone} of polar cones. Since \(\vct w\) is nonzero, we have the inclusion
  \begin{equation*}
    \Bigl\{\mtx Q \mid -C\cap \mtx Q D = \{\zerovct \}\Bigr\} \subset \Bigl\{ \mtx Q \mid C^\polar \cap \mtx Q D^\polar \ne \{\zerovct\}\Bigr\}.
  \end{equation*}

  For the other direction, suppose that  \(C^\circ \cap \mtx Q D^\circ \ne \{\zerovct\}\) for some rotation \(\mtx Q\in \orth{d}\).  By definition of polar cones, this implies the existence of a vector \(\vct w \ne \zerovct\) satisfying~\eqref{eq:sep1}---in other words, some nonzero vector weakly separates the cone \(-C\) from \(\mtx QD\).   We therefore find two alternatives: either \(-C\cap \mtx Q D = \{\zerovct\}\), or the closures \(-\overline{C}\) and \(\mtx Q \overline{D}\)  touch (cf.\ Definition~\ref{def:touch}).  In event notation, we have the inclusion
  \begin{equation*}
    \Bigl\{\mtx Q \mid C^\circ \cap \mtx Q D^\circ \ne \{\zerovct\} \Bigr\} \subset
    \Bigl\{\mtx Q\mid -C\cap \mtx Q D = \{\zerovct\}\Bigr\} \cup \Bigl\{\mtx Q \mid -\overline{ C} \text{ touches } \mtx Q \overline{D} \Bigr\}.
  \end{equation*}
But randomly oriented, nontrivial, closed cones touch with probability zero by Fact~\ref{fact:no-touching}, so the third set  above has measure zero. The conclusion follows by combining the two displayed inclusions. 
\end{proof}

\begin{proof}[Proof of Lemma~\ref{lem:polar-exact}]
For each \(i =1,\dotsc,n,n+1\), define   \(E_i\subset (\orth{d})^{n+1}\) and \(E_\polar\subset (\orth{d})^{n+1}\) by 
\begin{equation*}
  E_i:=\left\{-\mtx Q_i D_i  \cap  \sumnl_{j\ne i} \mtx Q_j D_j = \{\zerovct\} \right\} \qtq{and} E_\polar:=\bigl\{ \mtx Q_1 D_1^\polar \cap\dotsb\cap \mtx Q_{n}D_n^\polar\cap \mtx Q_{n+1} D_{n+1}^\polar\ne \{\zerovct\}\bigr\}.
\end{equation*}
With this notation, the statement of Lemma~\ref{lem:polar-exact} is equivalent to the claim
\begin{equation}\label{eq:EiEiO}
  \Prob\left( \bigcap\nolimits_{i=1}^{n+1} E_i\right) = \Prob( E_\polar).
\end{equation}
\noindent
 Let \(J\subset \{1,\dotsc, n,n+1\}\) be the set of indices \(j\) such that \(D_j \ne \{\zerovct\} \).  For any \(k \notin J\), the event \(E_k\) always occurs:
\begin{equation*}
\Prob( E_k)  = \Prob\left\{-\{\zerovct \}\cap \sumnl_{i\ne k } \mtx Q_k D_k = \{\zerovct\} \right\}= 1
\end{equation*}
because \(D_k =\{\zerovct\}\) for \(k\notin J\) by definition.  Therefore, 
\begin{equation}\label{eq:12}
  \Prob\left(\bigcap\nolimits_{j=1}^{n+1} E_j \right) = \Prob\left( \bigcap\nolimits_{j\in J} E_j\right).
\end{equation}
Note that this relation requires  that \(J\) is not empty, which holds true because we assume that \(D_j\ne \{\zerovct\}\) for at least two cones.  For each \(j\in J\),  both relations
\begin{equation*}
  D_j \ne \{\zerovct\} \qtq{and} \sumnl_{k\ne j} \mtx Q_k D_k \ne \{\zerovct\}
\end{equation*}
hold.  Indeed, the left-hand relation is the definition  of \(J\), while the right-hand relation follows because at least one of the remaining cones is nontrivial by assumption. From Proposition~\ref{prop:polar-prop},  for \(j\in J\),  the event \(E_j  \) is equal to \(E_\polar\) except on a set of measure zero.    Since finite unions and intersections of null sets are null, the intersection \(\bigcap_{j\in J} E_j \) is  equal to \(E_\polar\) except on a set of measure zero.  In particular,
\begin{equation*}
  \Prob\left( \bigcap\nolimits_{j\in J} E_j\right)= \Prob(E_\polar).
\end{equation*}
Combining this equality with~\eqref{eq:12} proves that~\eqref{eq:EiEiO} holds, which completes the claim.
\end{proof}

\section{The approximate kinematic formula}
\label{sec:proof-appr-kinem}

 The approximate kinematic formula is the main tool we use to derive the probability bounds in Theorem~\ref{thm:mult-bd}. This new formula extends the result~\cite[Thm.~7.1]{AmeLotMcC:13} to an arbitrary number of cones, and it incorporates  several technical improvements from the recent work~\cite{McCTro:13}. 

At its core, the approximate kinematic formula is based on an \emph{exact} kinematic formula for convex cones.  This kinematic formula is classical~\cite{San:76}, and the form we use here can be found, for example, in~\cite[Sec.~6.5]{SchWei:08}. Our derivation requires some background in conic integral geometry; we collect the relevant definitions and facts in Section~\ref{sec:conic-integr-geom}.  The proof of the approximate kinematic formula appears in Section~\ref{sec:proof-appr-kinem-1}.

\subsection{Background from conic integral geometry}
\label{sec:conic-integr-geom}

We start by defining the core parameters associated with convex cones.

\begin{definition}[Intrinsic volumes~\cite{McM:75}] Let \(C\in \cC_d\) be a polyhedral cone.  For each \(i=0,\dotsc,d-1, d\), the \(i\)th \emph{(conic) intrinsic volume} \(v_i(C)\) is equal to the probability that a Gaussian random vector projects into an \(i\)-dimensional face of \(C\), that is
  \begin{equation}\label{eq:int-vol-def}
    v_i(C) \defeq \Prob\Bigl\{ \Proj_{C}(\vct g) \in \relint(F_i)  \;\mid \;
       F_i \text{ is an \(i\)-dimensional face of } C
      \Bigr\}.
  \end{equation}
  This definition extends to all cones in \(\cC_d\) by approximation with polyhedral cones.%
\end{definition}

\noindent 
The next fact collects some basic facts about the intrinsic volumes.
\begin{fact}[Intrinsic volumes properties]
  For any closed, convex cone \(C \in \cC_d\), the following relations hold.
  \begin{enumerate}
  \item {\bf Probability.} The intrinsic volumes form a probability distribution:
    \begin{equation}
      \label{eq:intv-sum}
      \sumnl_{i=0}^d v_i(C) = 1\qtq{and} v_i(C) \ge 0.
    \end{equation}
  \item {\bf Polarity.} The intrinsic volumes reverse under polarity:
    \begin{equation}
      \label{eq:intv-polar}
      v_k(C) = v_{d-k}(C^\polar)
    \end{equation}
  \item {\bf Product.}  For any \(C' \in \cC_{d'}\), the  intrinsic volumes of the product \(C\times C'\)  satisfy
    \begin{equation}
      \label{eq:intv-prod}
      v_k(C\times C') = \sumnl_{i+j = k} v_i(C)v_j(C').
    \end{equation}
  \item {\bf Subspace.}  For an \(m\)-dimensional subspace \(L\subset \R^d\), we have
    \begin{equation}
      \label{eq:intv-subsp}
      v_k(L) =
      \begin{cases}
        1, &k=m, \\
        0,&\text{otherwise}.
      \end{cases}
    \end{equation}
  \end{enumerate}
\end{fact}
\noindent All of these facts appear in~\cite[Sec.~5.1]{AmeLotMcC:13}.  For future reference, we note here that~\eqref{eq:intv-prod} and~\eqref{eq:intv-subsp} together imply that for any \(C\in \cC_d\) and \(m\)-dimensional linear subspace \(L\), we have
\begin{equation}\label{eq:prod-w-subsp}
  v_k(C \times L ) = v_{k-m}(C)
\end{equation}
whenever \(k\ge m\).

Sums and partial sums of intrinsic volumes appear frequently in the theory of conic integral geometry, so we make the following definitions to simplify the later development. For any cone \(C\in \cC_d\) and index \(k= 0,\dotsc,d-1,d\), we define the \(k\)th \emph{tail-functional} \(t_k(C)\) by
\begin{align}
  \label{eq:tail-defn}
  t_k(C) &:= v_k(C) + v_{k+1}(C) + \dotsb = \sumnl_{j=k}^d v_k(C)
  \intertext{and the \(k\)th \emph{half-tail functional}}
  h_k(C) &:= v_k(C) + v_{k+2}(C)+\dotsb=\sum_{\substack{j=k\\j-k\text{ even}}}^d v_j(C).  \label{eq:half-tail-defn}
\end{align}
The tail functionals satisfy the following properties.
\begin{fact}[Properties of the tail functionals]%
  Let  \(C\in \cC_d\) be a closed, convex cone.
  \begin{enumerate}
  \item {\bf Gauss--Bonnet.} \textnormal{\cite[Eq.~(6.55)]{SchWei:08}}
    \begin{equation}\label{eq:guass-bonnet}
      2 h_1(C) =
      \begin{cases}
        0, & C \text{ an even-dimensional subspace} \\
        2, & C \text{ an odd-dimensional subspace} \\
        1, & \text{otherwise}
      \end{cases}
    \end{equation}
  \item {\bf Interlacing.} \textnormal{\cite[Prop.~5.7]{AmeLotMcC:13}} If \(C\) is not a linear subspace, then 
    \begin{equation}\label{eq:interlacing}
      h_k(C) \ge \frac{1}{2} t_k(C)\ge h_{k+1}(C) \qtq{for every} k=0,\dotsc,d-1,d.
    \end{equation}
  \item {\bf Duality.} \textnormal{\cite[Eq.~(6.9)]{AmeLotMcC:13}}
      We have the duality formula
      \begin{equation}\label{eq:tail-dual}
        t_k(C) = 1- t_{d-k+1}(C^\polar).
      \end{equation}
    \end{enumerate}
  \end{fact}

\subsubsection{Kinematic formulas}
\label{sec:kinematic-formulas}

For any two cones \(C,C' \in \cC_d \), the classical conic kinematic formula states~\cite[Eq.~(6.61)]{SchWei:08}
\begin{equation}\label{eq:kin-form}
  \Expect[ v_k( C \cap \mtx Q D )] = v_{d+k}(C \times D) \qtq{for} k = 1,\dotsc,d-1,d,
\end{equation}
where the expectation is over the random rotation \(\mtx Q\).  Note that our indices are shifted compared to the reference, and we have simplified the expression using the product rule~\eqref{eq:intv-prod}.   Using an inductive argument, we can extend this formula to the product of a finite number of cones.
\begin{fact}[Iterated kinematic formula]
  Let \(C_1,\dotsc,C_{n-1},C_n\in \cC_d\) be closed, convex cones and suppose that \(\mtx Q_1,\dotsc,\mtx Q_{n-1},\mtx Q_{n}\) are i.i.d.\ random rotations.  Then for all \(k=1,\dotsc,d-1,d\), we have
  \begin{equation}
    \label{eq:iter-conic-kin}
    \Expect[v_k(\mtx Q_1 C_1 \cap \dotsb \cap\mtx Q_{n-1}C_{n-1}\cap \mtx Q_n C_n)] = 
    v_{(n-1)d +k}(C_1\times\dotsb\times C_{n-1} \times C_n).
  \end{equation}
\end{fact}
The details are straightforward, so we refer to~\cite[Prop.~5.12]{McC:13} for the proof.    See~\cite[Thm.~5.13]{SchWei:08} for the analogous proof in the Euclidean setting.  A related fact is the following \emph{Crofton formula} for the probability that convex cones intersect nontrivially.
\begin{fact}[Iterated Crofton formula]
  \label{cor:iterated-croft} Let \(C_1,  \dotsc, C_{n-1}, C_n \in \cC_d\) be closed, convex cones, at least one of which is not a subspace.  Suppose \(\mtx Q_1,\dotsc,\mtx Q_{n-1}, \mtx Q_n\in \orth{d}\) are independent random rotations. Then
  \begin{equation}
    \label{eq:iter-crofton}
    \Prob\Bigl\{\mtx Q_1 C_1 \cap \dotsb\cap \mtx Q_{n-1}C_{n-1} \cap \mtx Q_n C_n \ne \{\zerovct\}\Bigr\}= 2 h_{(n-1)d+1}(C_1\times \dotsb \times C_{n-1}\times C_n).
  \end{equation}
\end{fact}
\noindent The proof, which appears in~\cite[Cor.~5.13]{McC:13}, simply combines the Gauss--Bonnet formula~\eqref{eq:guass-bonnet} with the kinematic formula~\eqref{eq:iter-conic-kin}.  The only obstacle involves verifying that the intersection of cones is  almost surely  not an odd-dimensional subspace so long as one of the cones in the intersection is not a subspace.  This technical point is proved in detail in~\cite[Lem.~5.13]{McC:13}.

\subsection{Proof of the approximate kinematic formula}
\label{sec:proof-appr-kinem-1}

The proof of Theorem~\ref{thm:mult-bd} begins with a  concentration inequality  for tail functionals. 
\begin{proposition}[Concentration of tail functionals] \label{prop:intv-conc} Let \(C_1, \dotsc,C_{n-1},C_n\in \cC_d\) and let \(\Omega\) and \(\theta \) be as  in~\eqref{eq:omega-theta}.
  Then for any \(\lambda >0\) and  integer \(k\ge\Omega + \lambda\), we have
  \begin{equation}
    t_{k}(C_1\times \dotsb \times C_{n-1}\times C_n) \leq  p_\theta(\lambda); \label{eq:tk+} 
  \end{equation}
\end{proposition}
\begin{proof}
  We follow the  argument of~\cite[Cor.~5.2]{McCTro:13}.   For any cone \(C \in \cC_d\),  we define the intrinsic volume random variable \(V_{C}\) on \(\{0,1,\dotsc, d\}\) by its  distribution:
  \begin{equation*}
    \Prob \{V_{C} = i\} = v_i(C).
  \end{equation*}
  The mean value of \(V_C\) is  equal to the statistical dimension, that is, 
  \begin{math}
    \Expect[V_C] =\sdim(C)
  \end{math}%
  ~\cite[Sec.~4.2]{McCTro:13}.  The product rule~\eqref{eq:intv-prod} for intrinsic volumes implies \(V_{C_1\times \dotsb \times C_{n-1}\times C_n}=  \sumnl_{i=1}^nV_{C_i}\) because the distribution of a sum of independent random variables is equal to the convolution of the distributions.  In particular,
  \begin{equation*}
    \Expect[ V_{C_1\times\dotsc\times C_{n-1}\times C_n}] = \sdim(C_1\times\dotsb\times C_{n-1}\times C_n) = \sumnl_{i=1}^n \sdim(C_i) =   \Omega.
  \end{equation*}
  
  With these facts in hands, we can complete the proof by tracing  the argument leading to~\cite[Cor.~5.2]{McCTro:13}.  The exponential moment of \(V_{C_1\times\dotsb\times C_{n-1}\times C_n}\)  factors as
  \begin{equation}\label{eq:moment-bd}
    \Expect \econst^{\zeta (V_{C_1\times \dotsc,\times C_n} - \Omega)} = \prod\nolimits_{i=1}^n \Expect \econst^{\zeta(V_{C_i} - \sdim(C_i))} \le \exp\left(\frac{\zeta^2 \theta^2}{1-2 \abs{\zeta}/3}\right) \qtq{for any} \abs{\zeta}<3/2,
  \end{equation}
  where the inequality follows from~\cite[Thm.~4.8]{McCTro:13} and the bound
  \begin{equation*}
    \frac{\econst^{2\zeta} - 2 \zeta - 1}{2} \le \frac{ \zeta^2}{1-2|\zeta|/3} \qtq{for all} \abs{\zeta}\le \frac{3}{2}.
  \end{equation*}
  Combining the moment bound~\eqref{eq:moment-bd} with the Laplace transform method  under the choice \(\zeta = \lambda/(2\theta^2 +2 \lambda/3)\) provides 
  \begin{equation*}
    t_{\lceil  \Omega + \lambda\rceil}(C_1 \times \dotsc \times C_n)  =  \Prob\{V_{C_1\times \cdots\times C_n} \ge \Omega + \lambda\} \le \exp\left(\frac{-\lambda^2/4}{\theta^2 + \lambda /3}\right).
  \end{equation*}
  The first equality above is the definition~\eqref{eq:tail-defn} of the tail functional.  Inequality~\eqref{eq:tk+}  follows because the integer \(k \ge \Omega +\lambda\) and the tail functionals are decreasing in \(k\). 
\end{proof}

\begin{proof}[Proof of Theorem~\ref{thm:appr-kin}]
  A simple dimension-counting argument shows that we incur no loss by assuming that at least one of the cones \(C_1,\dotsc,C_{n-1},C_n\) is not a subspace.  Indeed,  recall  from linear algebra that two generically oriented subspaces intersect nontrivially with probability zero if the sum of their dimensions is less or equal to the ambient dimension, but they intersect with probability one if the sum of their dimensions is greater than the ambient dimension.  When all of the cones are subspaces, the term \(\Omega\) is just the sum of the dimensions of the subspaces \(C_i\).  Evidently, when all of the cones are subspaces, the implications~\eqref{eq:appr-kin-upper} and~\eqref{eq:appr-kin-lower} hold with respective probability bounds zero and one.

 Suppose then that at least one of the cones is not a subspace.  For \(\Omega +m \le n d-\lambda\), the iterated kinematic formula~\eqref{eq:iter-crofton} bounds the probability of interest by
  \begin{align}
    \Prob\bigl\{\mtx Q_1 C_1\cap \dotsb\cap \mtx Q_n C_n \cap \mtx Q_{n+1} L \ne \{\zerovct\}\bigr\} 
    &= 2 h_{nd+1}(C_1 \times \dotsb\times C_{n}\times L)\label{eq:13} \\
    & \le t_{nd}(C_1 \times \dotsb\times C_{n}\times L),\notag 
  \end{align}
  where the inequality follows from the interlacing result~\eqref{eq:interlacing}.  Equation~\eqref{eq:prod-w-subsp}  and the upper tail bound~\eqref{eq:tk+} provides
  \begin{equation*}
    t_{nd}(C_1 \times \dotsb\times C_{n}\times L) = t_{nd-m}(C_1\times\dotsb\times C_{n})\le  p_\sigma(\lambda).
  \end{equation*}
  This completes the first claim~\eqref{eq:appr-kin-upper}.  

  The second claim follows along similar lines.  Suppose that \(\Omega +m \ge nd+\lambda\).  Combining the iterated kinematic formula~\eqref{eq:13} with the lower interlacing inequality~\eqref{eq:interlacing}, we see
  \begin{align}
    \Prob\bigl\{\mtx Q_1 C_1\cap \dotsb\cap \mtx Q_n C_n \cap \mtx Q_{n+1} L \ne \{\zerovct\}\bigr\} 
    &
    \ge t_{nd+1}(C_1\times \dotsb \times C_n\times L) \notag
     \\ 
     &= 1- t_d(C_1^\polar\times \dotsb \times C_{n}^\polar \times L^\perp),\label{eq:14}
  \end{align}
  where the final relation is~\eqref{eq:tail-dual}.  Using~\eqref{eq:prod-w-subsp} to shift the index of the tail functional, we find
  \begin{equation}\label{eq:15}
      t_d(C_1^\polar\times \dotsb \times C_{n}^\polar \times L^\perp) = t_m(C_1^\polar \times \dotsb\times C_n^\polar) \le p_\sigma(\lambda).
  \end{equation}
  The final inequality follows from the approximate kinematic formula~\eqref{eq:tk+}, which applies because 
  \begin{equation*}
    m \ge (nd - \Omega) + \lambda = \sumnl_{i=1}^n \sdim(C_i^\polar) + \lambda
  \end{equation*}
  by assumption and the polarity formula~\eqref{eq:sdim-polar}.  The final claim~\eqref{eq:appr-kin-lower} follows by combining~\eqref{eq:14} and~\eqref{eq:15}.
\end{proof}

\section{Degenerate case of the main theorem}
\label{sec:degenerate-case}

\begin{proof}[Proof of Theorem~\ref{thm:mult-bd} for the degenerate case.]
  We now consider the degenerate situation where all except possibly
  one of the cones \((D_i)_{i=1}^n\) is equal to the trivial cone
  \(\{\zerovct\}\). In this case, the restrictions in
  Lemma~\ref{lem:polar-exact} preclude using the polar optimality
  condition~\eqref{eq:polar-exact}. Instead, we study the success
  probability~\eqref{eq:9} directly.

  By our assumption, there is an index \(i_*\in\{1,\dotsc,n,n-1\}\)
  such that \begin{equation*}
  D_i = \{\zerovct\} \qtq{for all} i \in \{1,\dotsc,n,n+1\}\setminus\{ i_*\}.
\end{equation*}
This implies that 
\begin{equation*}
\text{either}\quad \mtx Q_i  D_i =\{\zerovct \} \quad \text{or} \quad \sumnl_{j\ne i } \mtx Q_j D_j = \{\zerovct\}
\end{equation*}
for every \(i=1,\dotsc,n,n+1\). Therefore, the
probability~\eqref{eq:9} is equal to one, so that \((\vct
x_i^\nat)_{i=1}^n\) is almost surely the unique optimal point of the
constrained demixing method~\eqref{eq:const-dmx} by
Lemmas~\ref{lem:achievable-stab} and~\ref{lem:rand-nullspace}.
Since~\eqref{eq:src} holds with the same probability
that~\eqref{eq:erc} holds under the random orientation model
(Lemma~\ref{lem:stab-exact-equiv}), we only need to verify that the
left-hand side of the implication~\eqref{eq:fail-mult-demix} never
holds.

By definition of \(\Delta\) and \(i_*\), we have \begin{equation}
  \Delta + d-m = \sumnl_{i=1}^{n+1} \sdim(D_i) =  \sdim(D_{i_*})  \le d
\end{equation}
because the statistical dimension is always less than the ambient dimension.  Rearranging, we find
\begin{equation*}
  m \ge \Delta > \Delta - \lambda_*
\end{equation*}
because \(\lambda_*>0\).  Hence, the left-hand side of the implication~\eqref{eq:fail-mult-demix} never holds.  
\end{proof}

\section{Numerical details}
\label{sec:numerical-details}

This section provides some specific numerical details of the experiments described in Section~\ref{sec:applications}.  

\paragraph*{Numerical environment.}
\label{sec:numerical-aspects}

All computations are performed using the \textsc{Matlab} computational platform. We generate i.i.d. rotations from the orthogonal group using the method described in~\cite{Mez:07}.  We solve~\eqref{eq:const-dmx} numerically using the \texttt{CVX} package~\cite{GraBoy:08,GraBoy:10} for  \textsc{Matlab}. All numerical precision settings are set at the default.   The empirical  level sets appearing in Figure~\ref{fig:sparse-sparse-sign} are determined using the \texttt{contour} function.

\paragraph*{Computing the statistical dimension.}\label{sec:comp-cons}  In order to draw the yellow curves in Figure~\ref{fig:sparse-sparse-sign}, we make use known statistical dimension computations.   The statistical dimension \(\sdim(\linf{\cdot},\vct x)= d/2\) whenever \(\vct x \in \{\pm 1\}^d\)  because the descent cone \(\Desc(\linf{\cdot},\vct x_3^\nat)\) is isometric to the positive orthant \(\R^d_+= \{\vct x \mid x_i\ge 0 \;\forall\; i = 1,\dotsc,d-1,d\}\) that has statistical dimension \(\sdim(\R^d_+) = d/2\)~\cite[Sec.~4.2]{AmeLotMcC:13}.

We estimate the statistical dimension of the descent cone of the \(\ell_1\) norm at sparse vectors  by solving the implicit formulas appearing in~\mbox{\cite[Eqs.~(4.12) \& (4.13)]{AmeLotMcC:13}} using \textsc{Matlab}'s \texttt{fzero} function.  These equations define a function  \(\psi\colon (0,1)\to (0,1)\) that satisfies
\begin{equation}\label{eq:psi-bd}
  \psi(k/d) - \frac{2}{\sqrt{kd}} \le \frac{\sdim(\lone{\cdot},\vct x)}{d} \le \psi(k/d)
\end{equation}
for every  vector \(\vct x\in \R^d\) with \(k\) nonzero elements.  The function \(d\cdot \psi\) thus provides and accurate approximation to the statistical dimension \(\sdim(\lone{\cdot},\vct x)\).

\paragraph*{Sparse, sparse, and sign} The experiment seen in Figure~\ref{fig:sparse-sparse-sign} [left] was conducted using the following procedure. Fix the ambient dimension \(d=200\) and for each sparsity \(k_1,k_2\in \{1,\dotsc,59,60\}\), we repeat the following steps \(25\) times:
\begin{enumerate}[itemsep=-2.5pt]%
\item Draw \(\mtx U_1\), \(\mtx U_2\), and \(\mtx U_3\) i.i.d.\ from the orthogonal group \(\orth{d}\).
\item For \(i=1,2\), generate independent sparse vectors \(\vct x_i\) with \(k_i\) nonzero elements by selecting the support  uniformly at random and setting each nonzero element \(\{\pm 1\}\) independently and with equal probability.
\item Draw \(\vct x_3^\nat \) by choosing each elements  from \(\{\pm 1\}\) independently and with equal probability.
\item Compute \(\vct z_0 = \mtx U_1 \vct x_1^\nat +\mtx U_2 \vct x_2^\nat +\mtx U_3 \vct x_3^\nat \).
\item Solve~\eqref{eq:sss-demix} for an optimal point \((\hvct x_i)_{i=1}^3\) using \texttt{CVX}.
\item Declare success if \(\linf{\hvct x_i -\vct x_i^\nat}< 10^{-5}\) for each \(i=1,2,3\).
\end{enumerate}
Figure~\ref{fig:sparse-sparse-sign} [left] shows the results of this experiment. The colormap indicates the empirical probability of success for each value of \(k_1\) and \(k_2\).  To compare this experiment to the guarantees of Theorem~\ref{thm:mult-bd}, we plot the curve \textsl{(yellow)} in \((k_1,k_2)\)-space such that
\begin{equation*}
 d\,\psi\left(\tfrac{k_1}{d}\right) + d\,\psi\left(\tfrac{k_2}{d}\right) = \frac{d}{2}
\end{equation*}
where \(\psi(k/d)\approx  \sdim(\lone{\cdot},\vct x)\) at \(k\)-sparse vectors \(\vct x\) (cf.~\eqref{eq:psi-bd}). Recalling that \(\sdim(\linf{\cdot},\vct x_3^\nat)=d/2\), we see that  the yellow curve in Figure~\ref{fig:sparse-sparse-sign} [left] shows the approximate center of the phase transition between success and failure predicted by Theorem~\ref{thm:mult-bd}.  It matches the empirical \(50\%\) success level set \textsl{(red)} very closely.

\paragraph*{Undersampled sparse and sparse}\label{sec:unders-sparse-sparse-1}
We fix the ambient dimension \(d=200\) and for each pair of sparsity levels \(k_1,k_2\in \{1,\dotsc,39,40\}\) and measurement number \(m \in \{25,50,75,100\}\), we repeat the following procedure \(35\) times:
\begin{enumerate}[itemsep=-2.5pt]
\item Draw the matrix \(\mtx A\in \R^{m\times d}\) with i.i.d.\ standard Gaussian entries and the rotations~\(\mtx U_1,\mtx U_2 \in \orth{d}\) i.i.d.\ from the orthogonal group~\(\orth{d}\).
\item Generate independent sparse vectors \(\vct x_1^\nat\) and \(\vct x_2^\nat\)  with \(k_1\) and \(k_2\) nonzero elements using the same method as above.
\item Compute \(\vct z_0 = \mtx A(\mtx U_1 \vct x_1^\nat + \mtx U_2 \vct x_2^\nat)\).
\item Solve~\eqref{eq:ssu} for an optimal point \((\hvct x_i)_{i=1}^2\) using \texttt{CVX}.
\item Declare success if \(\linf{\hvct x_i -\vct x_i^\nat}< 10^{-5}\) for \(i=1,2\).
\end{enumerate}
We present the results of this experiment in Figure~\ref{fig:sparse-sparse-sign} [right]. The colormap denotes the empirical probability of success for different values of \(k_1\) and \(k_2\), and each subpanel displays the results for a different value of \(m\).  Each subpanel also displays the curve \textsl{(yellow)} where
\begin{equation*}
  d\,\psi\left(\tfrac{k_1}{d}\right) + d\,\psi\left(\tfrac{k_2}{d}\right) = m.
\end{equation*}
The bound~\eqref{eq:psi-bd} guarantees that this curve is close to the theoretical phase transition predicted by Theorem~\ref{thm:mult-bd}. Again, we find very close agreement between the yellow curve and the empirical \(50\%\) success level set.

\section*{Acknowledgments}

MBM thanks Prof.\ Leonard Schulman for helpful conversations about this research. This research was supported by ONR awards N00014-08-1-0883 and N00014-11-1002, AFOSR award FA9550-09-1-0643, and a Sloan Research Fellowship.

\bibliographystyle{alpha} 
\footnotesize
\bibliography{bibliography}

\end{document}

%% file: macro.tex
\usepackage{fixltx2e} % Makes \( \) equation style robust, among other
\usepackage[T1]{fontenc}
\usepackage[utf8]{inputenc}
\usepackage{cmap} % Fixes ligatures, makes document searchable
\usepackage{graphicx}

\usepackage{setspace}

\usepackage{lmodern}%\usepackage{fouriernc}
\usepackage[scale=0.88]{tgheros}

\usepackage{bm} % boldmath must be called after the package

\usepackage{lettrine} % For big letters  at the beginning of sections
 \newcommand*{\Lettrine}[2]{%
   \lettrine[loversize=0.2,lines=2,nindent=0em]%
   {\textsf{#1}}{\textsf{#2}}}

\providecommand{\mathbold}[1]{\bm{#1}}

\usepackage{amsmath,amsbsy,amsgen,amscd,amsthm}
\DeclareMathAlphabet\mathbb{U}{msb}{m}{n}
\usepackage[cmsy]{MnSymbol} % Miscellaneous symbols like \llangle

\usepackage[centering,top=1.5in,bottom=1.2in,left=1.4in,right=1.4in]{geometry}
\usepackage{titling}
\usepackage{nopageno}
\pretitle{\noindent\rule{0.85\linewidth}{0.2mm}\par%
  \begin{raggedright}\LARGE\sffamily}
\posttitle{\par\end{raggedright}%
\noindent%
\rule{0.85\linewidth}{0.5mm}}

\preauthor{\par \vspace{0.5em}\large %
  \begin{tabular}[t]{c}\sffamily\itshape}
  \postauthor{\end{tabular}\par\thispagestyle{plain}}

\predate{ %
  \begin{tabular}[t]{c}%
      \sffamily\itshape}
\postdate{\end{tabular}\par}

\usepackage[runin]{abstract}

\usepackage[sf,bf,compact]{titlesec}

\usepackage{fancyhdr}
\renewcommand{\sectionmark}[1]%
{\markright{\MakeUppercase{\thesection.\ #1}}}

\usepackage{enumitem}
\setitemize{itemsep=0pt}
\setenumerate{itemsep=0pt}
\setlist[1]{labelindent=\parindent,font=\sffamily} % Recommended by enumitem package

\usepackage{booktabs,longtable,tabu} % Nice tables
\usepackage[font=small,margin=10pt,labelfont={sf,bf},labelsep={space}]{caption}

\usepackage[usenames,dvipsnames]{xcolor}
\definecolor{dark-gray}{gray}{0.3}
\definecolor{dkgray}{rgb}{.4,.4,.4}
\definecolor{dkblue}{rgb}{0,0,.5}
\definecolor{medblue}{rgb}{0,0,.75}
\definecolor{rust}{rgb}{0.5,0.1,0.1}
\usepackage{url}
\usepackage[colorlinks=true]{hyperref}
\hypersetup{linkcolor=dkblue}    
\hypersetup{citecolor=rust}      
\hypersetup{urlcolor=rust}     

\usepackage[final]{microtype}

\newtheoremstyle{myThm} % name
    {\topsep}                    % Space above
    {\topsep}                    % Space below
    {\itshape}                   % Body font
    {}                           % Indent amount
    {\sffamily\bfseries}                   % Theorem head font
    {.}                          % Punctuation after theorem head
    {.5em}                       % Space after theorem head
    {}  % Theorem head spec (can be left empty, meaning ‘normal’)

\newtheoremstyle{myRem} % name
    {\topsep}                    % Space above
    {\topsep}                    % Space below
    {}                   % Body font
    {}                           % Indent amount
    {\sffamily}                   % Theorem head font
    {.}                          % Punctuation after theorem head
    {.5em}                       % Space after theorem head
    {}  % Theorem head spec (can be left empty, meaning ‘normal’)

\newtheoremstyle{myDef} % name
    {\topsep}                    % Space above
    {\topsep}                    % Space below
    {}                   % Body font
    {}                           % Indent amount
    {\sffamily\bfseries}                   % Theorem head font
    {.}                          % Punctuation after theorem head
    {.5em}                       % Space after theorem head
    {}  % Theorem head spec (can be left empty, meaning ‘normal’)

\theoremstyle{myThm}
\newtheorem{theorem}{Theorem}[section]
\newtheorem{lemma}[theorem]{Lemma}
\newtheorem{proposition}[theorem]{Proposition}

\newtheorem{fact}[theorem]{Fact}

\newtheorem{bigthm}{Theorem}

\theoremstyle{myRem}

\theoremstyle{myDef}
\newtheorem{definition}[theorem]{Definition}

\let\originalleft\left
\let\originalright\right
\renewcommand{\left}{\mathopen{}\mathclose\bgroup\originalleft}
\renewcommand{\right}{\aftergroup\egroup\originalright}

\usepackage{mathtools}
\mathtoolsset{centercolon}

\numberwithin{equation}{section} 

\renewcommand{\phi}{\varphi}
\newcommand{\eps}{\varepsilon}

\renewcommand{\mid}{\mathrel{\mathop{:}}}

\newcommand{\nat}{{\smash{\natural}}}

\newcommand{\sumnl}{\sum\nolimits}

\newcommand{\defeq}{:=}

\newcommand{\qtq}[1]{\quad\text{#1}\quad} % "Quad-Text-Quad"

\DeclareMathOperator*{\minimizeOp}{minimize}
\newcommand{\subjectto}{\quad\text{\textnormal{subject to}}\quad}
\newcommand{\term}[1]{\emph{#1}}
\newcommand{\econst}{\mathrm{e}}

\newcommand{\zerovct}{\vct{0}} % Zero vector

\newcommand{\Id}{\mathbf{I}}

\newcommand{\ball}[1]{\mathsf{B}_{#1}}
\newcommand{\orth}[1]{\mathsf{O}_{#1}}

\providecommand{\mathbbm}{\mathbb} % In case we don't load bbm

\newcommand{\R}{\mathbbm{R}}

\newcommand{\polar}{\circ}

\newcommand{\abs}[1]{\left\vert {#1} \right\vert}

\newcommand{\argmin}{\operatorname*{arg\; min}}

\newcommand{\Prob}{\mathbbm{P}}

\newcommand{\Expect}{\operatorname{\mathbb{E}}}

\newcommand{\normal}{\mathrm{N}}

\newcommand{\vct}[1]{\mathbold{#1}}
\newcommand{\mtx}[1]{\mathbold{#1}}

\newcommand{\hvct}[1]{\hat{\vct{#1}}}
\newcommand{\tvct}[1]{\tilde{\vct{#1}}}

\newcommand{\transp}{t}

\newcommand{\psinv}{\dagger}

\newcommand{\range}{\mathcal{R}}
\newcommand{\nullity}{\mathcal{N}}

\newcommand{\Proj}{\ensuremath{\bm{\Pi}}} % Projection operator

\newcommand{\norm}[1]{\left\Vert {#1} \right\Vert}

\newcommand{\lone}[1]{\norm{#1}_{{\ell_1}\!}}

\newcommand{\linf}[1]{\norm{#1}_{\ell_\infty}}
\newcommand{\sone}[1]{\norm{#1}_{S_1}}

\newcommand{\enorm}[1]{\norm{#1}}

\newcommand{\enormsq}[1]{\enorm{#1}^2}

\newcommand{\cC}{\mathcal{C}} % Set of nonempty, closed convex cones
\newcommand{\Desc}{\mathcal{D}}

\newcommand{\sdim}{\delta}
\DeclareMathOperator{\relint}{relint}

\newcommand{\pInd}{\mathbf{1}} % Indicator function for probability